\documentclass[a4paper]{article}
\usepackage{graphicx}
\usepackage{psfrag}
\usepackage{color}
\usepackage[small]{caption} 
\usepackage{makeidx}
\usepackage{amsmath}
\usepackage{amsfonts}
\usepackage{amsthm}
\usepackage{setspace}
\usepackage{relsize}
\usepackage{subfigure}
\usepackage{listings}
\usepackage{hyperref}
\usepackage{amssymb}
\usepackage{authblk}
\newtheorem{proposition}{Proposition}[section]
\newtheorem{theorem}[proposition]{Theorem}

\newtheorem{definition}[proposition]{Definition}
\newtheorem{remark}[proposition]{Remark}

\allowdisplaybreaks[4]
\numberwithin{equation}{section}
\numberwithin{figure}{section}

\newcommand{\reali}{{\mathbb{R}}}
\newcommand{\interi}{{\mathbb{Z}}}
\newcommand{\naturali}{{\mathbb{N}}}
\renewcommand{\epsilon}{\varepsilon}
\renewcommand{\phi}{\varphi}
\renewcommand{\theta}{\vartheta}

\renewcommand{\d}[1]{\mathinner{\mathrm{d}{#1}}}

\author{Claudio Canuto}
\author{Anna Cattani}

\affil{Dipartimento di Scienze Matematiche, Politecnico di Torino,\\ Corso D. Abruzzi, 24, 10129 Torino, Italy}

\title{The derivation of continuum limits of neuronal networks with gap-junction couplings}
\date{April 30, 2013}
\begin{document}

\maketitle

\begin{abstract}
  \smallskip
  We consider an idealized network, formed by $N$ neurons individually described by the FitzHugh-Nagumo equations
and connected by electrical synapses. The limit for $N \to \infty$ of the resulting discrete model is thoroughly
investigated, with the aim of identifying a model for a continuum of neurons having an equivalent behaviour. Two strategies
for passing to the limit are analysed: i) a more conventional approach, based on a fixed nearest-neighbour
connection topology accompanied by a suitable scaling of the diffusion coefficients; ii) a new approach, in which the number of connections to any given neuron varies with $N$ according to a precise law, which simultaneously
guarantees the non-triviality of the limit and the locality of neuronal interactions.  Both approaches yield in the limit 
a pde-based model, in which the distribution of action potential obeys a nonlinear reaction-convection-diffusion equation; 
convection accounts for the possible lack of symmetry in the connection topology. Several convergence issues are
discussed, both theoretically and numerically.
\end{abstract}

\section{Introduction}
The computer simulation of the behaviour of complex networks with a huge number of nodes, such as networks of neurons in
some portion of the brain, is a formidable challenge. The intrinsic difficulties of such a task may be alleviated to some extent by
identifying one or more multiscale structures within the networks; this allows one to describe and simulate different scales by different models, while posing the problem of the interaction among scales. 

Within a multiscale framework, the co-existence of
discrete and continuous models is a natural option, which may lead to significant savings. Higher-level nodes, or interactions, may be affordably given an individual description (e.g., by a system of coupled ordinary differential equations), whenever their number is small to moderate. On the contrary, this approach would be computationally prohibitive for
the description of lower-level nodes or interactions, if their number is exceedingly large. In this case, a possible alternative may consist  in modelling the huge population of individuals by a continuum, confined in some spatial region, and describe its
behaviour by a limited number of variables, e.g., submitted to satisfy partial differential equations. One recognizes here a process underlying the mathematical description of many physical phenomena, e.g., in Fluid Dynamics.

The derivation of a continuum model may be accomplished by ``passing to the limit'' in a discrete model, assuming that the
number of individuals tends to infinity. The present paper aims at investigating one such limit process. Specifically,
the discrete model we start from is inspired by the connections of electrical (rather than chemical) nature within a neuronal network. In order to better motivate our model, we provide now some biological background.


Despite in the past years almost all networks have been represented as constituted by neurons that are interconnected by chemical synapses, electrical synapses are largely present in the nervous system.
In the sequel, we will use indifferently the terms \textit{electrical synapses} and \textit{gap junctions}. However, for the sake of completeness, gap junctions are the morphological equivalent of electrical synapses. In particular, as specified in~\cite{KeenerSneyd}, gap junctions exist between near-neighbour neurons and they allow low-resistance electrical transmissions. Indeed, at an electrical synapse a current $I_{\rm gap}$ is generated which is proportional to the difference between the action potentials $v$ of the post-synaptic and pre-synaptic neurons (see, e.g.,~\cite{Ermentrout} eq. (7.12));
explicitly, we have for some $d>0$
\begin{equation}\label{eq:gapjunction}
I_{\rm gap} = d (v_{\rm post} - v_{\rm pre}) \;.
\end{equation}
This establishes a \textit{diffusive coupling} between neighbouring neurons.

Unfortunately, the analysis of electrical synapses \textit{in situ} presents severe technical difficulties and therefore their specific roles are still largely unexplored. Nevertheless, in the past ten years, the topic concerning gap junction networks has been object of several investigations, sometimes leading to paradoxical results (see e.g.~\cite{Galarreta, Marder, Yu}). 

In order to build up the sample network we will consider, several ingredients are taken into account.
First of all, we model each single cell as an excitable element by exploiting the FitzHugh-Nagumo model (see~\cite{Fitzhugh}). The excitable feature means that neurons may not fire intrinsically without any synaptic inputs. Furthermore, each cell belongs to the same functional class, avoiding the presence of heterogeneity. This agrees with authors in~\cite{Galarreta} who stress that electrical synapses exist exclusively between neurons of a specific class. In particular, despite many works underline the presence of electrical synapses between inhibitory neurons (see e.g.~\cite{Galarreta}), the existence of electrical connections between excitatory neurons is demonstrated in the early postnatal stages (see~\cite{Yu}).
Finally, as we will specify, we consider the presence of both bidirectional (non-rectifying) and unidirectional (rectifying) synapses as claimed in~\cite{Marder}.

We now describe the content of this paper. After setting our mathematical model of an idealized neuronal network with
electrical-type coupling between neurons, we carefully investigate the ``passage to the limit'' as the number of neurons
tends to infinity, while they remain confined in a fixed and bounded spatial region. We identify two different manners of
increasing the population of the network so that a non-trivial continuum limit is obtained. The first one assumes a fixed 
topology of the network (nearest-neighbour connections) but makes the proportionality coefficient in~\eqref{eq:gapjunction} to depend upon the total number of neurons according to a specific law; conversely, the second manner keeps this
coefficient fixed but suitably increases the number of connections per neuron. Both methods lead to equivalent continuous
models, in which the action potential is the solution of a reaction-diffusion partial differential equation
 (or a reaction-convection-diffusion
equation if connections are not symmetric, i.e., if rectifying synapses are allowed). 
Our arguments apply in any spatial dimension,
although we detail them in 1D and we sketch their extension to 2D. Clear numerical evidence confirms all theoretical results.
At last, an example of random connections is also presented.

\section{The FitzHugh-Nagumo model of a single neuron}
The FitzHugh-Nagumo model~\cite{Fitzhugh} was introduced as a dimensional reduction of the well-known Hodgkin-Huxley model~\cite{HodgkinHuxley}. It extracts the Hodgkin-Huxley fast-slow phase plane and presents it in a simplified form. The resulting model is more analytically and numerically tractable and it maintains a certain biophysical meaning.
The model is constituted by two equations in two variables $v$ and $r$. The first one is the fast variable called \textit{excitatory}: it represents the transmembrane voltage. The second variable is the slow \textit{recovery variable}: it describes the time dependence of several physical quantities, such as the electrical conductance of the ion currents across the membrane.
The FitzHugh-Nagumo equations, using the notation in~\cite{Murray}, are given by:
\begin{equation}\begin{aligned}
\label{Eq:FN}
\dot{v}&=-v(a-v)(1-v)-r=:f(v,r)\;,\\
\dot{r}&=bv-cr=:g(v,r)\;,
\end{aligned}\end{equation}
where, $a,\,b,\,c\in\mathbb{R}^+$ are parameters of the model.
The model describes neurons as excitable elements which have two key properties. Firstly, they are characterized by their excitability behaviour: a sufficiently large stimulus provokes a very large response, that is, a small perturbation to the quiescent state of a neuron can provoke a large excursion of its potential. Secondly, they are characterized by their refractoriness: the elements cannot be excited during the period which follows the stimulus.

Throughout the paper, following~\cite{Wallisch}, we will adopt the FitzHugh-Nagumo model with $a=0.25$, $b=0.001$, $c=0.003$, to describe the behaviour of each neuron in our network.

\section{Diffusive coupling within the network}
We suppose that our network contains $N$ neurons, identified by integer labels $i=1,\cdots ,N$; labels may refer to the physical position of the neurons, but other ways to index neurons could be used in a more convenient way.
Electrical-type connections in the neuronal network are easily described by basic concepts from graph theory (see, e.g.,~\cite{BapatKalitaPati}).
Let us consider a graph $G=(V,E)$, where $V=\{1,\cdots ,N\}\subset\naturali$ is the set of vertices and $E\subset V\times V$ is the set of edges.
The so-called \textit{adjacency matrix} $A_G=(a_{ij})$ is an $N\times N$ matrix whose entries are:
\begin{displaymath}
   a_{ij}=
	\begin{cases} w_{ij} & \mbox{if }(i,j)\in E(G)\\
				  0 & \mbox{else}\;,
	\end{cases}
\end{displaymath} where $i,\,j=1,\cdots ,N$ and the weights are strictly positive.

Exploiting the adjacency matrix, and assuming the gap-junction law~\eqref{eq:gapjunction} for the interaction between adjacent neurons, we define the FitzHugh-Nagumo model with diffusive coupling as follows:
   \begin{equation}\begin{aligned}
	\dot{v}_i&= f(v_i,r_i)+\sum_{j\neq i} a_{ij}(v_j-v_i)\;,\\[-5pt]
	\dot{r}_i&= g(v_i,r_i).
	\label{Eq:FHNiNew}
	\end{aligned}\end{equation}
Specifically, the summation describes the influence on the $i-$th neuron of all neurons linked to it; it produces a diffusion effect within the network. The simplest example is given by the expression in~\eqref{Eq:DiffCoup}, which models nearest-neighbour interactions in a chain of neurons.

Introducing the diagonal \textit{degree matrix} $D_G:=\mbox{diag}(d_i)$ with $d_i=\sum_{j\neq i}a_{ij}$, and the \textit{Laplacian matrix} $L_G:=D_G-A_G=(l_{ij})$, the previous system can be written as 
\begin{displaymath}\begin{aligned}
	\underline{\dot{v}}&=\underline{f}(\underline{v},\underline{r})-L_G\underline{v}\;,\\
	\underline{\dot{r}}&=\underline{g}(\underline{v},\underline{r})\;,
\end{aligned}\end{displaymath}
where $\underline{v}=(v_i)$, $\underline{r}=(r_i)$ and $\underline{f}(\underline{v},\underline{r})=(f(v_i,r_i))$, $\underline{g}(\underline{v},\underline{r})=(g(v_i,r_i))$.

In the sequel, we assume that all weights $w_{ij}$ are equal and precisely $w_{ij}=d$ for some $d>0$, which we will call the \textit{diffusion coefficient}. Let us introduce the set $\mathcal{Q}(i)$ of all indices $q$ such that neuron $i+q$ is linked to neuron $i$, i.e., $a_{i,i+q}\neq 0$.
Then, the model~\eqref{Eq:FHNiNew} can be written as 
\begin{equation}\begin{aligned}
	\dot{v}_i&=f(v_i,r_i)+d\sum_{q\in \mathcal{Q}(i)}(v_{i+q}-v_i)\;,\\[-5pt]
	\dot{r}_i&=g(v_i,r_i)\;.
	\label{Eq:FHNiQ}
\end{aligned}\end{equation} 
In most cases, we shall consider $ \mathcal{Q}(i)= \mathcal{Q}$ independent of $i$, thus assuming a homogeneous network
topology.

We are interested in describing the behaviour of the network as the number of neurons increases, identifying conditions on the model which lead to non-trivial asymptotic patterns in the limit $N\rightarrow\infty$.
We assume that the network is contained in a bounded region $B$ (independent of $N$) of the Euclidean space $\reali ^m$, for some $1\leq m\leq 3$; let us denote by $x_i\in B$ the physical position of the $i-$th neuron. Then, we assume that the distance of any point $\hat{x}\in B$ from the network tends to zero as $N\rightarrow\infty$, and the distance of each neuron from its neighbours in the network has a similar behaviour.

If interactions between neurons are local, we can give an expression of the diffusive term in~\eqref{Eq:FHNiQ} which is based on the Taylor expansion of the differences $\Delta v_{i,q}=v_{i+q}-v_i$. Precisely, let us assume that at each time there exists a sufficiently smooth function $v$ defined in $B$ such that $v_i=v(x_i)$ for $i=1,\cdots ,N$. Then, setting $\Delta x_{i,q}=x_{i+q}-x_i$, we have
\begin{equation}
\Delta v_{i,q}=\nabla v(x_i)\Delta x_{i,q}+\frac{1}{2}\Delta x_{i,q}^T Hv(x_i)\Delta x_{i,q}+\mbox{h.o.t.}\;,
\label{Eq:TaylorGen}
\end{equation}
where $\nabla v$ denotes the gradient vector of $v$, whereas $Hv$ denotes the Hessian matrix of $v$.
Substituting this expression into~\eqref{Eq:FHNiQ}, we obtain a representation of the diffusive term by which we can find the conditions on the coefficient $d$ and/or the sets $\mathcal{Q}(i)$ (depending on the network) yielding a non-trivial limit as $N\rightarrow\infty$. We will detail our analysis assuming a specific distribution of neurons in the one-dimensional case first, and then we will consider the multi-dimensional extension.
   
\section{One-dimensional dynamics}
We consider neurons disposed over a closed chain, i.e., a ring. Each neuron occupies a specific physical position $x_i$ in the interval $B=[0,1]$ given by
\begin{equation}
x_i=(i-1)\Delta x=\frac{i-1}{N}\quad\mbox{with}\ 1\leq i\leq N\;,
\label{Eq:Mesh}
\end{equation}
where $N$ is the number of elements equally distributed along the chain and, consequently, $\Delta x=1/N$ is the distance between any two adjacent ones. Since the chain is closed, we assume periodic boundary conditions, i.e., we set $v_0=v_N$ and $v_{i+kN}=v_i$ for any $k\in\interi$.

\subsection{Nearest-neighbour interactions}
\label{NNI}
Let us first consider two symmetric nearest-neighbour interactions for each neuron. This translates in considering the set of connections per neuron $\mathcal{Q}(i)=\mathcal{Q}=\{\pm 1\}$.  In this case, the diffusive coupling assumes the following form:
\begin{equation}
-(L_Gv)_i=d\sum_{q=\pm 1}(v_{i+q}-v_i)=d[(v_{i+1}-v_i)+(v_{i-1}-v_i)]=d(v_{i+1}-2v_i+v_{i-1})\;.
\label{Eq:DiffCoup}
\end{equation}
An interesting dynamics produced by~\eqref{Eq:FHNiQ}, which will represent a test case for the subsequent discussion, is obtained by applying an initial stimulus to the central neuron ($i=N/2$, assuming $N$ even) of the line. Specifically, its action potential is initially set to the value $2$, whereas all the other variables are set to $0$. Considering the diffusion coefficient $d=0.05$ (see~\cite{Wallisch}), the resulting dynamics is constituted by two pulses that travel in opposite directions in the whole set of neurons (see e.g.~\cite{KeenerSneyd} for the analysis of travelling pulses). A sample dynamics is shown in Figure~\ref{Fig:Links2d1}. We observe for further reference that a similar dynamics is obtained starting from an initial stimulus of the action potential given by a Gaussian function concentrated around the central neuron. In all cases, at the end of dynamics, neurons return to the quiescent state. In fact, neurons are modelled as excitable units and then, after the excitation, they undergo a long refractory period. In this period they are blind to any stimulus. This is the reason why two travelling pulses that collide depress their signals.

\begin{figure}[t!]
  \centering
  \subfigure{\includegraphics[width=0.7\textwidth]{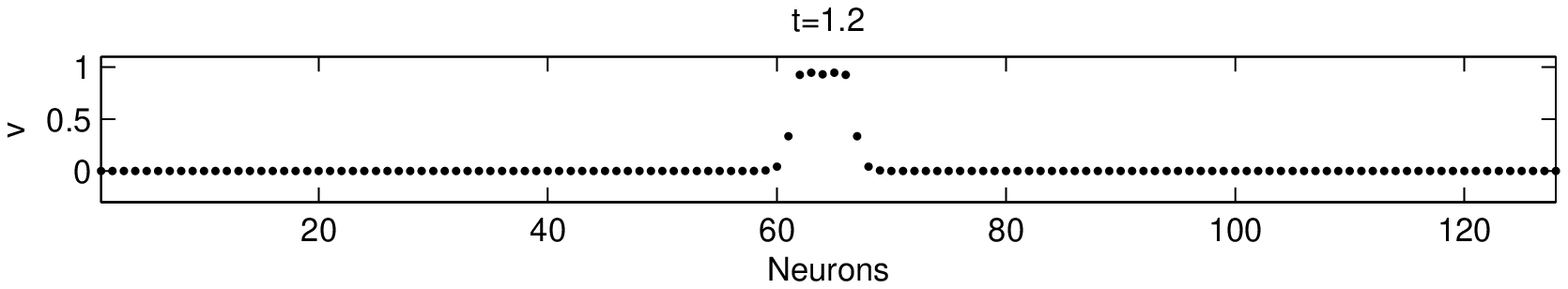}}\\\vspace{-0.65cm}
  \subfigure{\includegraphics[width=0.7\textwidth]{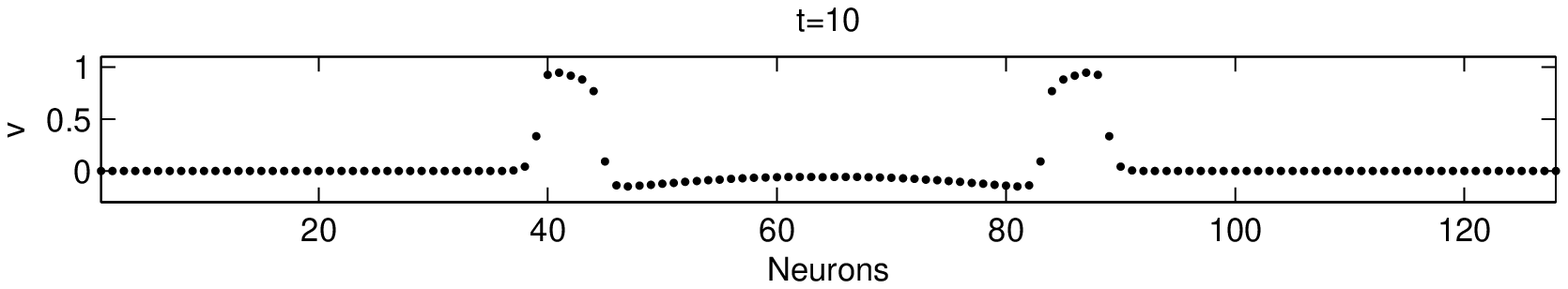}}\\
 \caption{\label{Fig:Links2d1} Propagation of an initial pulse along a closed ring of $N=128$ neurons}
\end{figure}

We now focus on how the dynamics produced by our model depends upon $N$. The first observation is that, if the diffusion parameter $d$ is kept fixed, then the diffusive effect tends to vanish as $N\rightarrow\infty$. This can be seen in two ways. On the one hand, if neurons get close to each other and the action potential varies in a smooth manner, then the differences on the right-hand side of~\eqref{Eq:DiffCoup} tend to zero, implying the vanishing of the diffusion term $L_Gv$ in each node. On the other hand, considering the test case introduced above, it is easily seen that the effect of, say, doubling $N$ is equivalent to have a chain of neurons with the same spacing but with double length; this means that on the original chain, waves have half the length and propagate with half the speed.

In order to obtain non-trivial diffusion effects in the limit, one possibility - that we call Approach I - consists in letting the parameter $d$ grow with $N$, i.e., $d=d_N$. The precise dependence can be found by exploiting the Taylor expansion~\eqref{Eq:TaylorGen}, which in the present setting becomes
\begin{equation}
\Delta v_{i,q}=q\Delta x v^\prime(x_i) + \frac{1}{2}q^2\Delta x^2 v^{\prime\prime}(x_i) + \mbox{h.o.t.}\;,
\label{Eq:Taylor}
\end{equation}
where the prime indicates differentiation with respect to the spatial variable $x$.
Therefore, the following expression holds for the diffusive term:
\begin{equation}
-(L_Gv)_i=d_N\left[(v_{i+1}-v_i)+(v_{i-1}-v_i)\right]=d_N[\Delta x^2 v^{\prime\prime}(x_i) +\mbox{h.o.t.}]\;.
\label{Eq:diffApproachOne}
\end{equation}
We choose $d_N$ in such a way that $d_N\Delta x^2$ is independent on $N$, say
\begin{equation}\label{Eq:def-dN}
d_N\Delta x^2=d^\ast
\end{equation}
for a fixed constant $d^\ast >0$. Hence, we obtain 
\begin{equation}
d_N=\frac{d^\ast}{\Delta x^2}=d^\ast N^2\;,
\label{Eq:dN}
\end{equation}
i.e., $d_N$ is proportional to the square of the number of neurons. The fact that $d_N$ is proportional to $N^2$ is not surprising: the spectral gap of the Laplacian matrix has the same behaviour as $1/(N^2)$.

As $N\rightarrow\infty$, the discrete model 
\begin{equation}\begin{aligned}
	\dot{v}_i&=f(v_i,r_i)+d_N\left[(v_{i+1}-v_i)+(v_{i-1}-v_i)\right]\;,\\
	\dot{r}_i&=g(v_i,r_i)\;,
	\label{Eq:DiscreteModel}
\end{aligned}\end{equation}
``converges'' to a continuous model. To support this statement, we observe that the quantity $\mbox{h.o.t.}$ in~\eqref{Eq:diffApproachOne} is given by
\begin{equation}
\nonumber
\mbox{h.o.t.}=\frac{1}{12}\Delta x^4 v^{(iv)}(\bar{x}_i)\;,
\end{equation}
where $\bar{x}_i\in (x_{i-1},x_{i+1})$ and $v^{(iv)}$ is assumed continuous in $[0,1]$. Thus, we have
\begin{equation}\label{eq:TaylorExpansionAppI}
\begin{split}
-(L_Gv)_i&=\frac{d^\ast}{\Delta x^2} \left[(v_{i+1}-v_i)+(v_{i-1}-v_i)\right] \\
&=d^\ast v^{\prime\prime}(x_i)+\frac{d^\ast}{12}\Delta x^2 v^{(iv)}(\bar{x}_i)\;.
\end{split}
\end{equation}
It follows that if we fix any point $\hat{x}\in [0,1]$ and, for each $N$, we consider a neuron of index $i=i(N)$ such that
\begin{equation}
\nonumber
x_{i(N)}=\frac{i(N)}{N}\rightarrow\hat{x}\quad \mbox{as }N\rightarrow\infty\;,
\end{equation}
then,
\begin{equation}
\nonumber
\lim_{N\rightarrow\infty}d_N \hspace{-0.1cm}\sum_{q\in\mathcal{Q}}(v_{i(N)+q}-v_{i(N)})=d^\ast v^{\prime\prime}(\hat{x})\;.
\end{equation}
We conclude that a continuum of neurons is the results of the limit process of letting $N\rightarrow\infty$, and 
\begin{equation}
\label{Eq:LimitSym}
\begin{array}{l}
\displaystyle\frac{\partial v}{\partial t}= f(v,r)+d^\ast\frac{\partial ^2 v}{\partial x^2}\;,\\[10pt]
\displaystyle\frac{\partial r}{\partial t}=g(v,r)\;,
\end{array}
\end{equation}
is the system of nonlinear partial differential equations of incomplete parabolic type which describes the action potential and the recovery variable in the whole set of neurons. Note that the first equation is similar to the so-called cable equation, which describes the distribution of the potential along the axon of a single neuron (see, e.g.~\cite{Ermentrout,Scott}). 
Reaction-diffusion models like~\eqref{Eq:LimitSym} are studied e.g. in~\cite{Fitzhugh1968mot,Smoller}.

We observe that the discrete model~\eqref{Eq:dN}--\eqref{Eq:DiscreteModel} can be viewed as a numerical semi-discretization (in space) of the PDE system~\eqref{Eq:LimitSym}, obtained by using a second-order centered finite difference method on the equally-spaced~\eqref{Eq:Mesh}. Thus, if the solution of~\eqref{Eq:LimitSym} is sufficiently smooth as in the case of an initial Gaussian stimulus, we expect to have quadratic convergence in $\Delta x$ of the discrete solutions, at any fixed time $t>0$, as it can be deduced from the fact that the error term on the right-hand side of~\eqref{eq:TaylorExpansionAppI} is proportional to $\Delta x^2$.

We now give an example. Following the choice of parameters presented in~\cite{Wallisch}, we set $d=0.05$ and we consider the case $N=128$ as a reference one, i.e., we impose $d_N=d$ for $N=128$, which yields
\begin{equation}
d^\ast=\frac{0.05}{128^2}=3.0518\cdot 10^{-6}\;.
\label{Eq:dstar}
\end{equation}
A comparison of several discrete solutions is presented in Figure~\ref{Fig:ConfrS1}. The (b) plots clearly document the convergence of the discrete dynamics towards a limit one.  
Note that these dynamics are generated by applying an initial stimulus ${v_i}_{|t=0}=2$ to a number of neurons proportional to $N$ around the center of the chain; in the limit, the initial action potential takes the value $2$ in an interval of positive length symmetrically placed around the point $x=1/2$, and vanishes elsewhere.
\begin{figure}
\centering
\begin{minipage}{\textwidth}
\centering
  \subfigure{\includegraphics[width=0.65\textwidth]{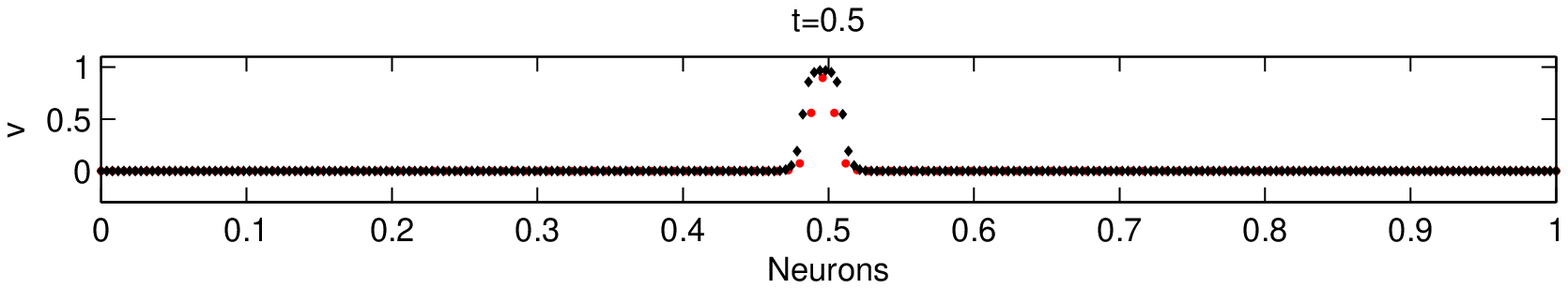}}\\\vspace{-0.65cm}
  \subfigure{\includegraphics[width=0.65\textwidth]{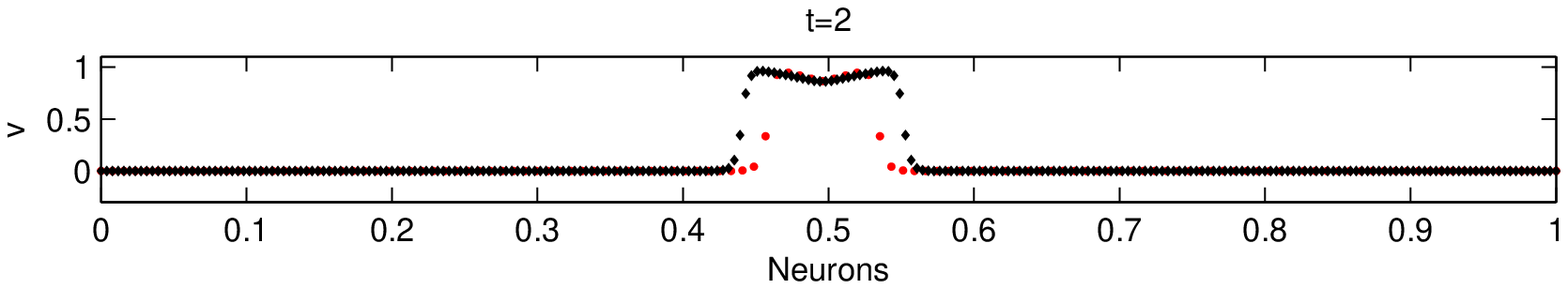}}\\\vspace{-0.65cm}
  \subfigure{\includegraphics[width=0.65\textwidth]{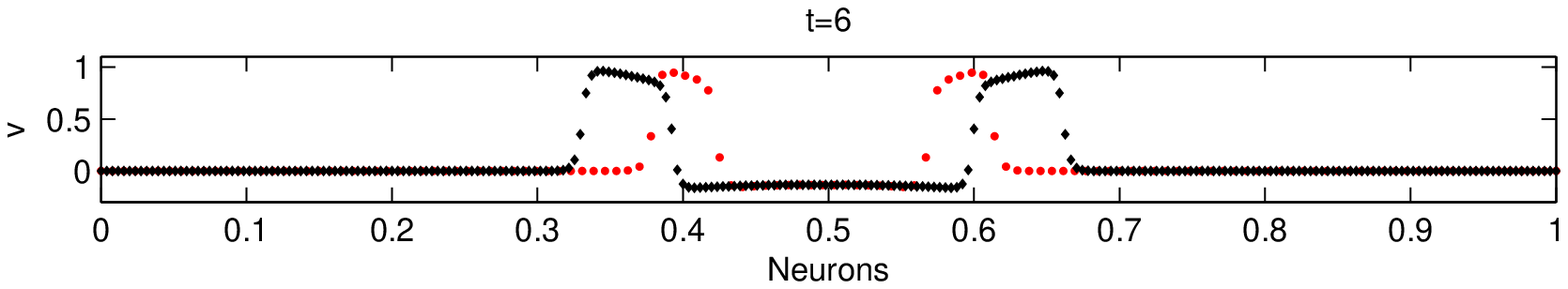}}\\\vspace{-0.65cm}
  \setcounter{subfigure}{0}
  \subfigure[]{\includegraphics[width=0.65\textwidth]{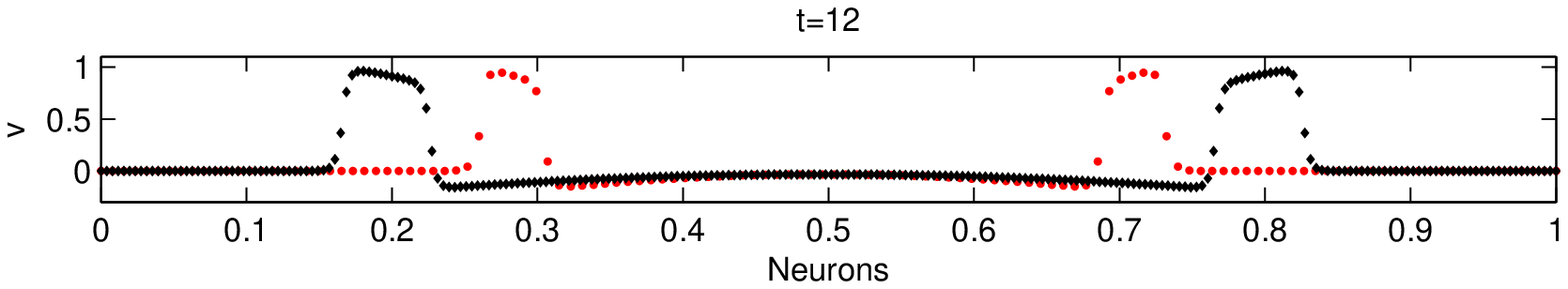}}
\end{minipage}
\begin{minipage}{\textwidth}
\centering
  \subfigure{\includegraphics[width=0.65\textwidth]{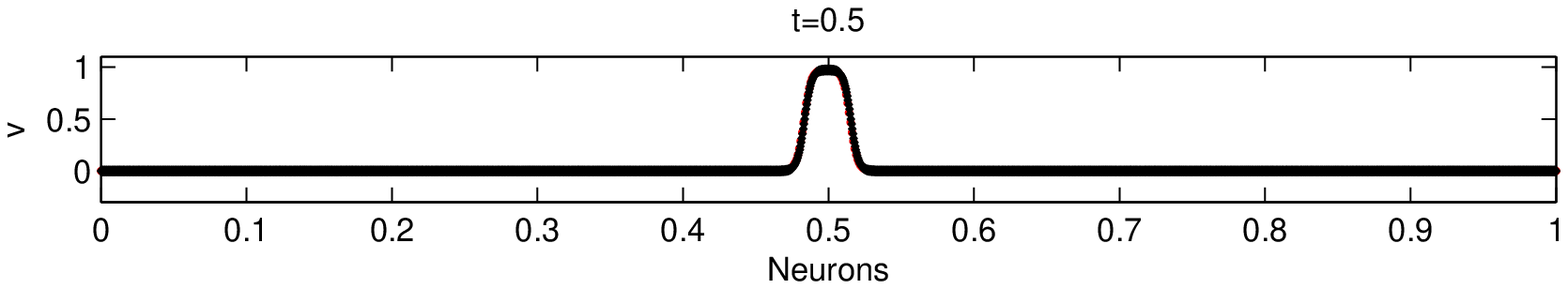}}\\\vspace{-0.65cm}
  \subfigure{\includegraphics[width=0.65\textwidth]{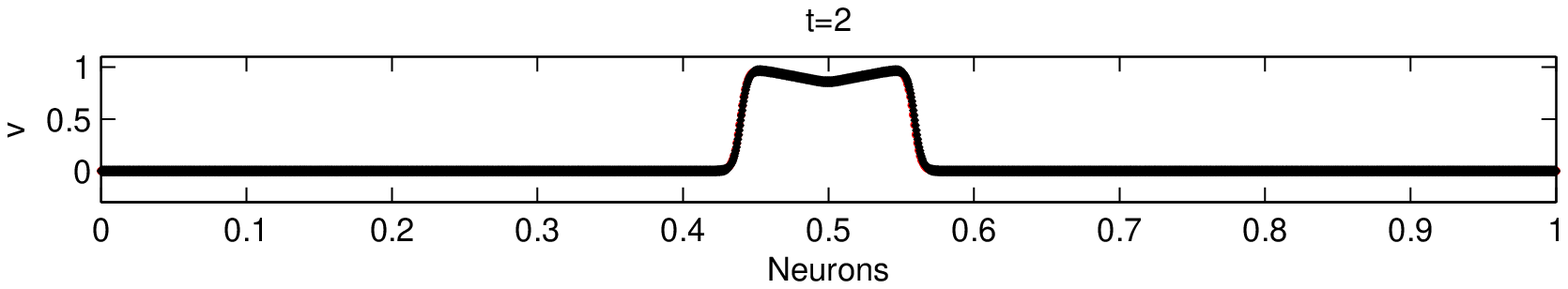}}\\\vspace{-0.65cm}
  \subfigure{\includegraphics[width=0.65\textwidth]{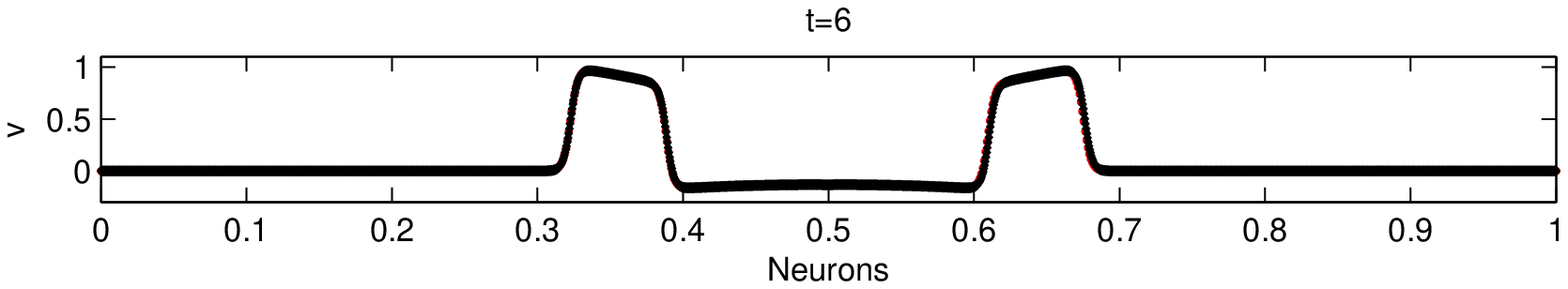}}\\\vspace{-0.65cm}
    \setcounter{subfigure}{1}
  \subfigure[]{\includegraphics[width=0.65\textwidth]{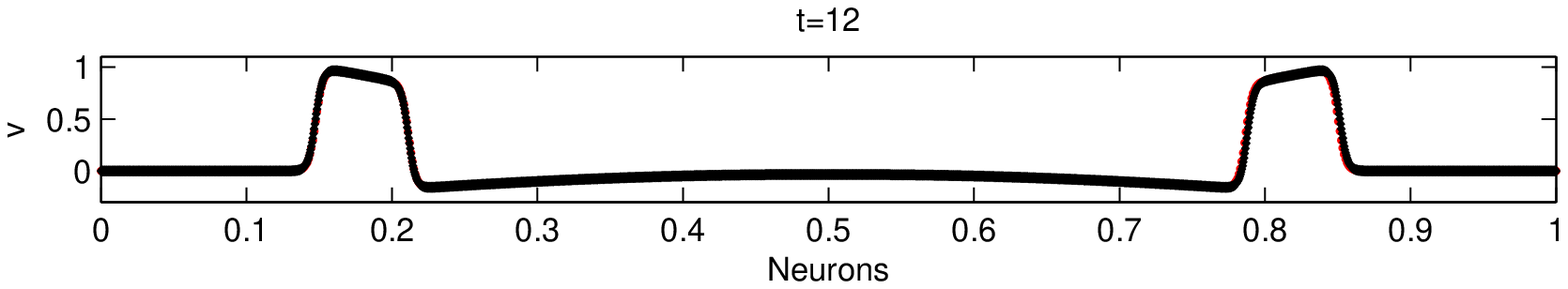}}
\end{minipage}
\caption{\label{Fig:ConfrS1} 
Convergence of the discrete model~\eqref{Eq:DiscreteModel}-\eqref{Eq:dN} (Approach I) as $N \to \infty$. Evolution 
of pulses (a) for $N=128$ (red dots) and $N=256$ (black dots), (b) for $N=1024$ (red dots) and $N=2048$ (black dots)
}
\end{figure}

\begin{remark}\label{rem:var-coeff}
A more general situation considers $d=d_N$ in~\eqref{Eq:DiffCoup} also depending on $i$ and $q$, i.e., 
the diffusive coupling law is replaced by
\begin{equation}
-(L_Gv)_i=d_{i+1}(v_{i+1}-v_i)+d_{i-1}(v_{i-1}-v_i)\;,
\label{Eq:DiffCoupModif}
\end{equation}
where $d_{i\pm 1}=N^2 d^*(x_{i\pm 1})$ and $d^*$ is a smooth function. In the limit,  the diffusion term in~\eqref{Eq:LimitSym} is replaced by $\frac{\partial}{\partial x}\left(d^\ast\frac{\partial v}{\partial x} \right)$. For simplicity, we confine ourselves to the constant-coefficient case.
\end{remark}

\subsection{Extended range interactions}
\label{Sec:Extended1D}
We now introduce a second approach to reproduce the same limit dynamics emerged above, which avoids rescaling the diffusion coefficient with the square of the number of neurons. This alternative way - which we call Approach II - consists of increasing the number of connections per neuron according to a specific law (and just slightly adjust the diffusion coefficient).

Since the core idea is to consider a number of connections per neuron that varies as a function of $N$, let us define the following set:
\begin{equation}\label{eq:QN}
\mathcal{Q}=\mathcal{Q}_N=\{\pm 1,\pm 2,\cdots ,\pm Q_N\}\;,
\end{equation}
where $Q_N$ is a positive integer to be determined.
Thus, neurons linked to the $i$-th one belong to the interval 
\begin{equation}
\label{Eq:I}
I=[x_i-Q_N\Delta x,x_i+Q_N\Delta x]\;.
\end{equation}
Using again Taylor expansions, the sum in the diffusive coupling becomes
\begin{equation}
\nonumber
\sum_{q\in\mathcal{Q}}(v_{i+q}-v_i)=\left(\sum_{q=1}^{Q_N} q^2\right)\Delta x^2 v^{\prime\prime}(x_i)+\mbox{h.o.t.}\;.
\end{equation}
Introducing the function $\varphi :\reali^+\rightarrow\reali^+$ defined as
\begin{equation}
\label{Eq:phi}
\phi (x)=\frac{x(x+1)(2x+1)}{6}
\end{equation}
%
and invoking the identity 
\begin{equation}
\nonumber
\sum_{q=1}^{n} q^2=\phi(n)\quad \forall n\geq 1\;,
\end{equation}
we obtain, 
\begin{equation}
\label{Eq:DiffCoupER}
-(L_Gv)_i=d\sum_{q\in\mathcal{Q}_N}(v_{i+q}-v_i)=d[\phi(Q_N)\Delta x^2 v^{\prime\prime}(x_i)+\mbox{h.o.t.}]\;.
\end{equation}
We would like to choose ${Q}_N$ in such a way that 
\begin{equation}
d\frac{\phi(Q_N)}{N^2}=d^\ast\;,
\label{Eq:FindQd}
\end{equation}
for a fixed constant $d^\ast >0$. This equation admits a unique solution, say $Q_N^r$, which however need not be an integer. Therefore, we choose $Q_N$ as the nearest integer to $Q_N^r$.

\begin{proposition}
\label{Property}
The number of neurons linked to any given one grows proportionally to the power $N^{2/3}$ of the total number of neurons. Indeed\footnote{
For any two  non-negative sequences $A_N$ and $B_N$, we will use the symbols
\begin{eqnarray*}
A_N \sim B_N &\iff& A_N/B_N \to 1 \text{ for } N \to \infty \;, \\
A_N \simeq B_N &\iff& c B_N\leq A_N \leq c' B_N\quad\mbox{with } c,c'>0 \;, \\
A_N\lesssim B_N &\iff& A_N\leq cB_N\quad\mbox{with }c>0.
\end{eqnarray*}
},
\begin{equation}
\nonumber
Q_N\sim Q^r_N\sim\left(\frac{3d}{d^\ast}N^2 \right)^{1/3}=c N^\frac{2}{3}\;.
\end{equation}
\end{proposition}
\begin{proof}
By definition, $Q_N^r$ satisfies
\begin{equation}
d\frac{\phi(Q_N^r)}{N^2}=d^\ast\;.
\label{Eq:FindQdr}
\end{equation}
The result follows recalling that $\phi (x)\sim \frac{x^3}{3}$ for $x\rightarrow\infty$.
\end{proof}

Let us underline that, although the number of neurons linked to any given one grows with $N$, interactions remain local, i.e., these neurons belong to a neighbourhood whose size decays with $N$. Indeed, considering the $i-$th neuron and recalling~\eqref{Eq:I}, we have
\begin{equation}
|I|\simeq Q_N\Delta x\simeq N^{-1/3}\;.
\label{Eq:ILaw}
\end{equation}
Thus, we expect that the limit model, as $N\rightarrow\infty$, be again described by partial differential equations.

As specified above, the slight shift from $Q_N^r$ to $Q_N$ provokes the necessity of slightly modifying the diffusion coefficient. Precisely, we define $d_N$ so that the identity
\begin{equation}
d_N\frac{\varphi (Q_N)}{N^2}=d^\ast\\
\label{Eq:dast}
\end{equation}
is satisfied. An alternative possibility, which will be explored later on and which leads to similar effects, would be to define $d_N^\ast$ so that
\begin{equation}
d\frac{\varphi (Q_N)}{N^2}=d_N^\ast\;.
\label{Eq:dNast}
\end{equation}
The coefficient $d_N$ is really a small perturbation of $d$, as the next proposition indicates.
\begin{proposition}
\label{PropdNd}
Let $d_N$ be the diffusion coefficient defined in~\eqref{Eq:dast}. Then, one has
\begin{equation}
\nonumber
|d_N-d|\lesssim N^{-\frac{2}{3}}\;.
\end{equation}
\end{proposition}
\begin{proof}
From~\eqref{Eq:FindQdr} and~\eqref{Eq:dast}, we obtain the following equality:
\begin{equation}
\label{Eq:Equal}
d_N\frac{\phi(Q_N)}{N^2}=d\frac{\phi(Q_N^r)}{N^2}\;.
\end{equation}
Since $Q_N$ is defined as the nearest integer to $Q_N^r$,
\begin{equation}
\label{Eq:QnQnr}
|Q_N^r-Q_N|\leq \frac{1}{2}\;,
\end{equation}
and then, $Q_N=Q_N^r+\varepsilon_N$ with a proper choice of $\varepsilon_N$, such that $|\varepsilon_N|\leq 1/2$.
Writing $\phi(Q_N^r)=\phi(Q_N)+\phi(Q_N^r)-\phi(Q_N)$
and substituting in~\eqref{Eq:Equal}, we get 
\begin{equation}
|d_N-d|=\frac{|\phi(Q_N^r)-\phi(Q_N)|}{\phi(Q_N)}d\;.
\end{equation}
Using~\eqref{Eq:QnQnr} and omitting computations, we conclude that $|\phi(Q_N^r)-\phi(Q_N)|\lesssim N^{4/3}$ while $\phi(Q_N)\simeq N^{2}$. 
This gives the desired estimate.
\end{proof}

In order to obtain the continuous model as a limit of the discrete model for $N\rightarrow\infty$, we observe that, if the fourth derivative of $v$ is continuous in $[0,1]$, the diffusion term~\eqref{Eq:DiffCoupER} takes the form
\begin{equation}\label{eq:TaylorExpansionAppII}
\begin{aligned}
-(L_Gv)_i&=d_N\left[\varphi (Q_N)\Delta x^2 v^{\prime\prime}(x_i)+\mbox{h.o.t.} \right]\\
&=d^\ast v^{\prime\prime}(x_i)+\frac{d_N}{12}\Delta x^4\sum_{q=1}^{Q_N}q^4 v^{(iv)}(\bar{x}_{i,q})\;,
\end{aligned}
\end{equation}
where $\bar{x}_{i,q}$ are suitable points in the interval $(x_{i-q},x_{i+q})$. Since $\sum_{q=1}^{Q_N}q^4\sim \frac{1}{5}Q_N^5$, using Property~\ref{Property} and Proposition~\ref{PropdNd}, we deduce that 
\begin{equation}\label{eq:errorBehaviour}
d_N\Delta x^4\sum_{q=1}^{Q_N}q^4\simeq N^{\frac{10}{3}-4}=N^{-\frac{2}{3}}\rightarrow 0\mbox{ as }N\rightarrow\infty\;.
\end{equation}

Therefore, proceeding as in Section~\ref{NNI}, if we fix any point $\hat{x}\in [0,1]$ and we consider a neuron of index $i=i(N)$ such that $x_{i(N)}\rightarrow\hat{x}$ as $N\rightarrow\infty$, we conclude that
\begin{equation}
\nonumber
{-(L_Gv)}_{i(N)}\rightarrow d^\ast v''(\hat{x})\mbox{ as }N\rightarrow\infty\;. 
\end{equation}
This means that Approach II yields in the limit the same system~\eqref{Eq:LimitSym} of partial differential equations, that we got from Approach I.

We now illustrate the asymptotic behaviour of the quantities defined above, for the same test case considered in the previous subsection.
We choose again $d=0.05$, and we enforce that for $N=N_0=128$ we have $Q_{N_0}=Q_{N_0}^r=1$, which corresponds to the nearest-neighbour interaction previously considered; we also enforce $d_{N_0}=d$, and consequently we get
\begin{equation}
\nonumber
d^\ast =\frac{d}{N_0^2}\;,
\end{equation}
which is precisely~\eqref{Eq:dstar}. Increasing $N$ by powers of $2$, i.e., setting $N=N_02^p$ with $p\geq 1$, the algorithm presented above produces the values of $Q_N$ and $d_N$ shown in Table~\ref{Tab:QandD}. The last column of this table, as well as Figure~\ref{Fig:diffdNd} (left), quantitatively support the asymptotic estimates proven in Propositions~\ref{Property} and~\ref{PropdNd}. Some representative dynamics obtained with Approach II are documented in Figure~\ref{ConfrS2}; they should be compared to those given in Figure~\ref{Fig:ConfrS1}. The evolutions of the action potentials produced by the discrete model with $N=1024$, and by a very accurate solution of the continuous model~\eqref{Eq:LimitSym} are documented in Figure~\ref{ConfrS12_1024}. While the shapes of the pulses are already well captured, their speed of propagation is less accurately 
reproduced; this should be related to the fourth-order error term on the right-hand side of~\eqref{eq:TaylorExpansionAppII},
whose decay is slower than in Approach I as indicated by~\eqref{eq:errorBehaviour} compared to~\eqref{eq:TaylorExpansionAppI}.

\begin{table}
\centering
\caption{Number of connections per neuron $Q_N$, diffusion coefficient $d_N$ and relative error as a function of $N=N_0 2^p$ with $N_0=128$.}
{\relsize{-2}
\begin{tabular}{rrrlc}
$p$ & $N=N_0 2^p$ & $Q_N$ & $d_N$ & $|d_N-d|/d$\\
\hline
$0$ & $128$ \rule{0pt}{2.5ex} \rule[-1ex]{0pt}{0pt} & 1 & 0.0500 & $0$\\
\hline
$1$ & $256$ \rule{0pt}{2.5ex} \rule[-1ex]{0pt}{0pt} & 2 & 0.0400 & 2.0$\cdot 10^{-1}$\\
\hline
$2$ & $512$ \rule{0pt}{2.5ex} \rule[-1ex]{0pt}{0pt} & 3 & 0.0571 & 1.4$\cdot 10^{-1}$\\
\hline
$3$ & $1024$ \rule{0pt}{2.5ex} \rule[-1ex]{0pt}{0pt} & 5 & 0.0582 & 1.6$\cdot 10^{-1}$\\
\hline
$4$ & $2048$ \rule{0pt}{2.5ex} \rule[-1ex]{0pt}{0pt} & 9 & 0.0490 & 1.0$\cdot 10^{-1}$\\
\hline
$5$ & $4096$ \rule{0pt}{2.5ex} \rule[-1ex]{0pt}{0pt} & 14 & 0.0504 & 8.7$\cdot 10^{-3}$\\
\hline
$6$ & $8192$ \rule{0pt}{2.5ex} \rule[-1ex]{0pt}{0pt} & 23 & 0.0473 & 5.3$\cdot 10^{-2}$\\
\hline
$7$ & $16384$ \rule{0pt}{2.5ex} \rule[-1ex]{0pt}{0pt} & 36 & 0.0505 & 1.1$\cdot 10^{-2}$\\
\hline
$8$ & $32768$ \rule{0pt}{2.5ex} \rule[-1ex]{0pt}{0pt} & 58 & 0.0491 & 1.8$\cdot 10^{-2}$\\
\hline
$9$ & $65536$ \rule{0pt}{2.5ex} \rule[-1ex]{0pt}{0pt} & 92 & 0.0496 & 6.3$\cdot 10^{-3}$\\
\hline
$10$ & $131072$ \rule{0pt}{2.5ex} \rule[-1ex]{0pt}{0pt} & 146 & 0.0500 & 4.9$\cdot 10^{-4}$\\
\hline
$11$ & $262144$ \rule{0pt}{2.5ex} \rule[-1ex]{0pt}{0pt} & 232 & 0.0500 & 1.2$\cdot 10^{-3}$\\
\hline
$12$ & $524288$ \rule{0pt}{2.5ex} \rule[-1ex]{0pt}{0pt} & 369 & 0.0500 & 2.3$\cdot 10^{-3}$\\
\hline
$13$ & $1048576$ \rule{0pt}{2.5ex} \rule[-1ex]{0pt}{0pt} & 586 & 0.0500 & 2.1$\cdot 10^{-3}$\\
\hline
$14$ & $2097152$ \rule{0pt}{2.5ex} \rule[-1ex]{0pt}{0pt} & 930 & 0.0500 & 4.3$\cdot 10^{-4}$\\
\hline
$15$ & $4194304$ \rule{0pt}{2.5ex} \rule[-1ex]{0pt}{0pt} & 1476 & 0.0500 & 7.4$\cdot 10^{-4}$\\
\hline
$16$ & $8388608$ \rule{0pt}{2.5ex} \rule[-1ex]{0pt}{0pt} & 2344 & 0.0500 & 1.6$\cdot 10^{-4}$\\
\hline
$17$ & $16777216$ \rule{0pt}{2.5ex} \rule[-1ex]{0pt}{0pt} & 3721 & 0.0500 & 3.0$\cdot 10^{-5}$\\
\hline
$18$ & $33554432$ \rule{0pt}{2.5ex} \rule[-1ex]{0pt}{0pt} & 5907 & 0.0500 & 2.3$\cdot 10^{-5}$\\
\hline
$19$ & $67108864$ \rule{0pt}{2.5ex} \rule[-1ex]{0pt}{0pt} &  9377 & 0.0500 & 2.9$\cdot 10^{-7}$\\
\hline
$20$ & $134217728$ \rule{0pt}{2.5ex} \rule[-1ex]{0pt}{0pt} &  14885 & 0.0500 & 7.1$\cdot 10^{-5}$\\
\end{tabular}
}
\label{Tab:QandD}
\end{table}

\begin{figure}[htbp]
\includegraphics[width=0.45\textwidth]{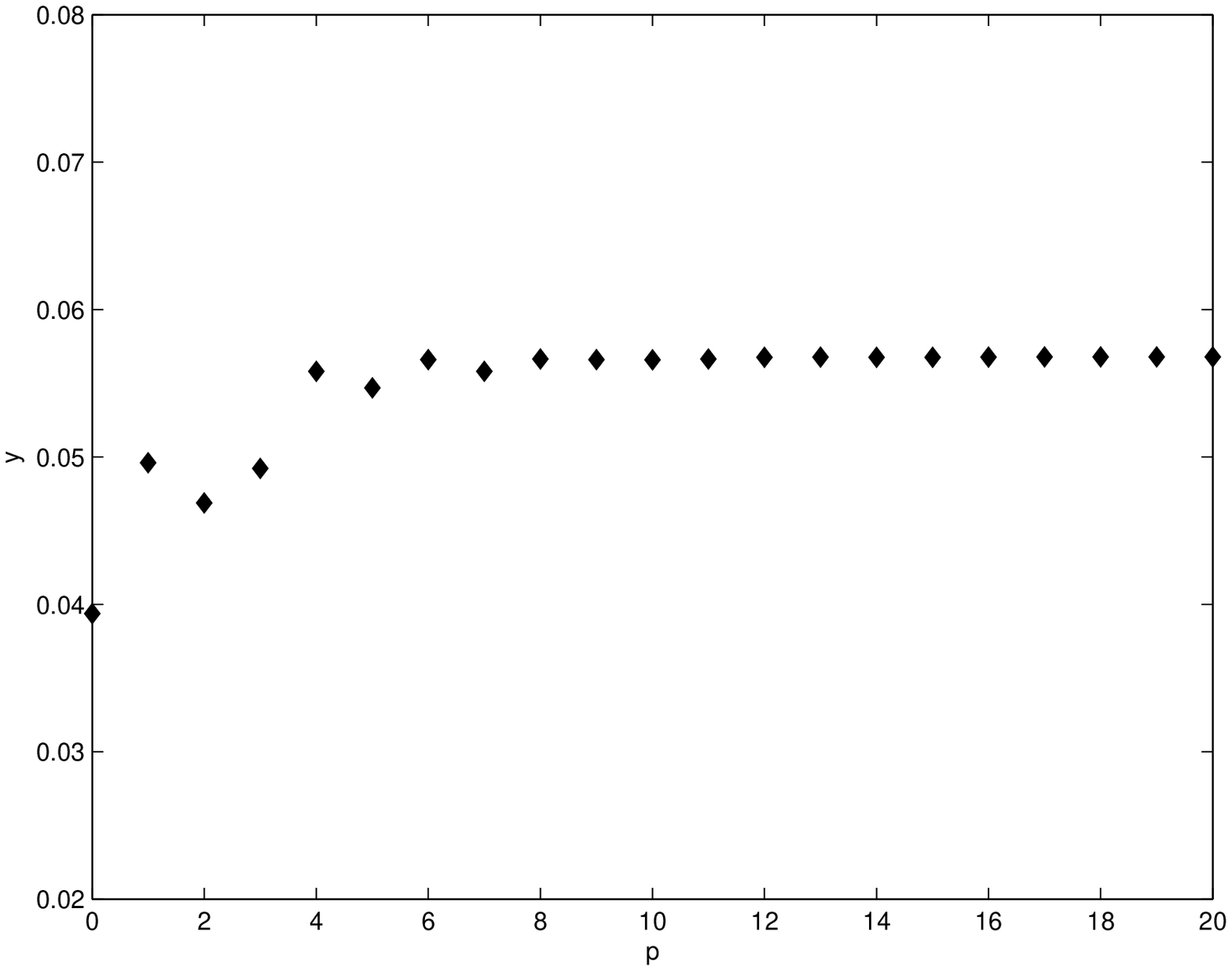}
\includegraphics[width=0.45\textwidth]{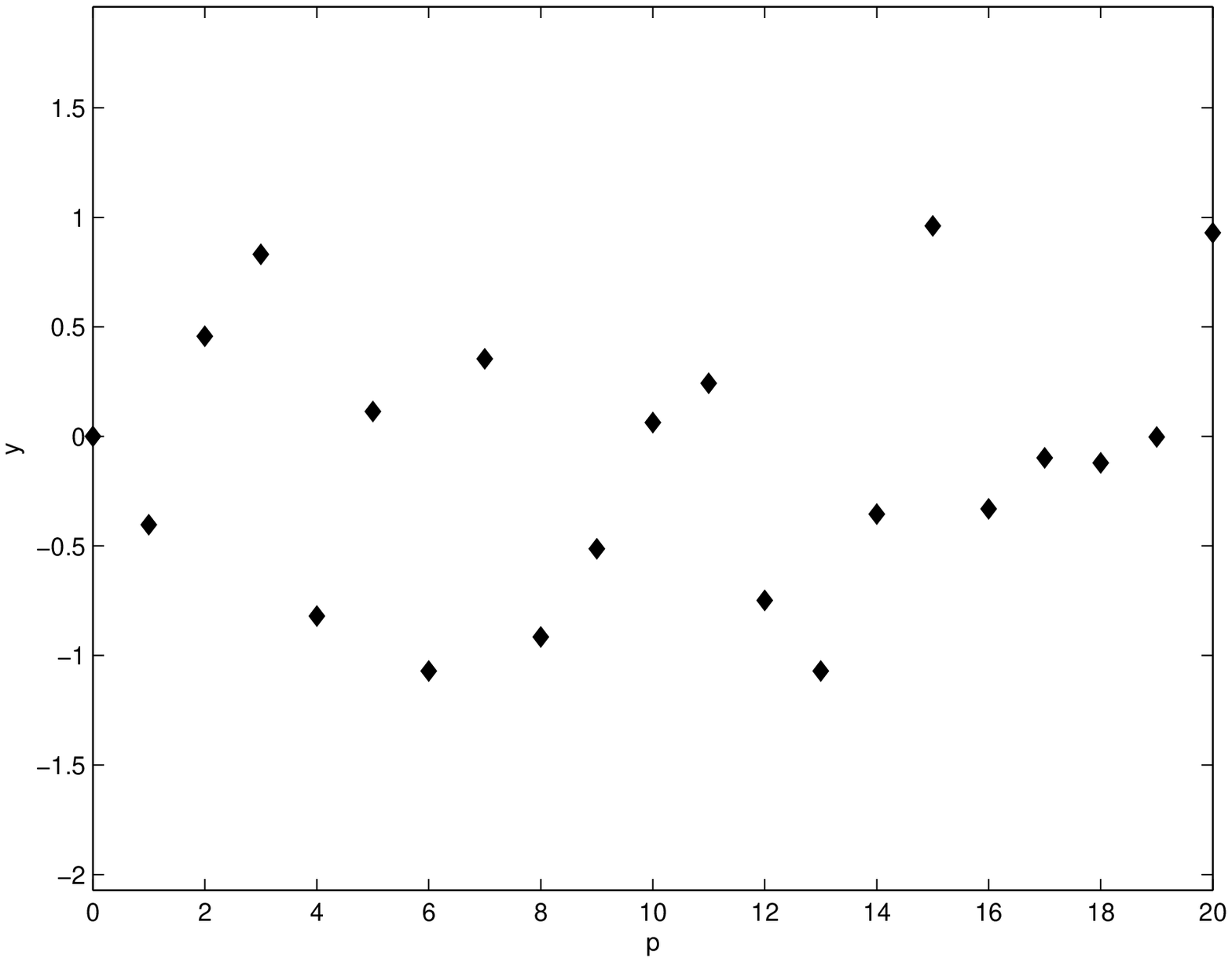}
\caption{\label{Fig:diffdNd} Plots of $Q_N/N^{2/3}$ (left) and $(d_N-d)/N^{-2/3}$ (right) vs $p$, where $N=N_02^p$}
\end{figure}

\begin{figure}
\centering
\begin{minipage}{\textwidth}
\centering
  \subfigure{\includegraphics[width=0.65\textwidth]{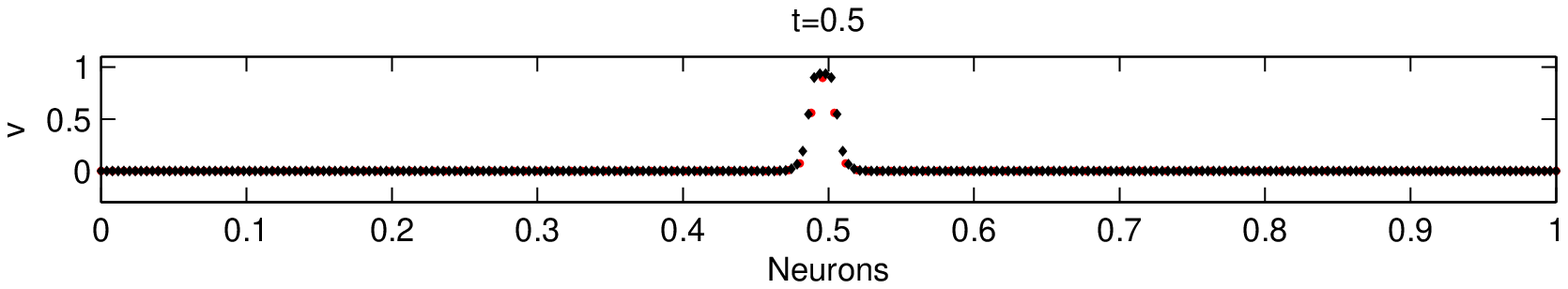}}\\\vspace{-0.65cm}
  \subfigure{\includegraphics[width=0.65\textwidth]{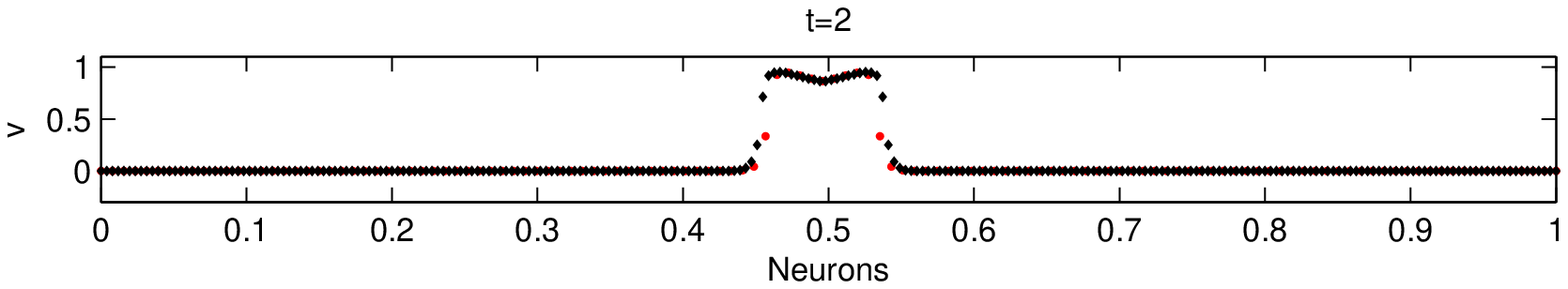}}\\\vspace{-0.65cm}
  \subfigure{\includegraphics[width=0.65\textwidth]{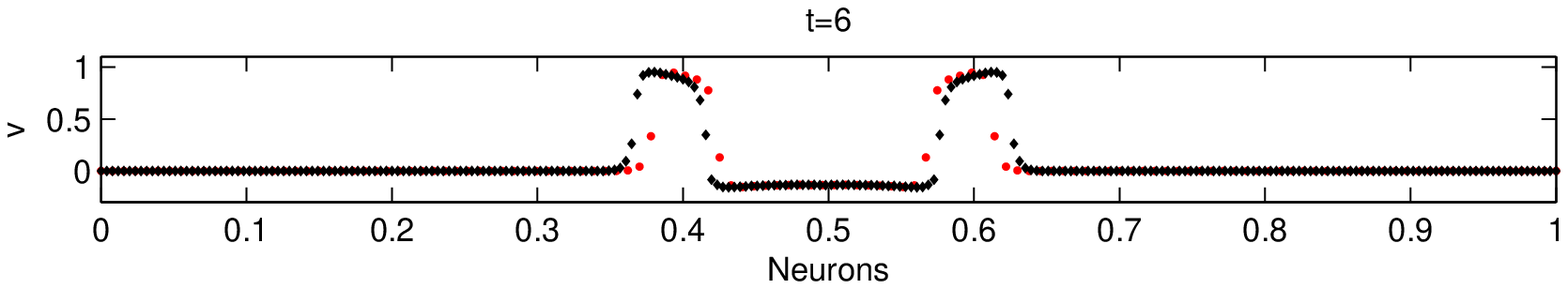}}\\\vspace{-0.65cm}
  \setcounter{subfigure}{0}
  \subfigure[]{\includegraphics[width=0.65\textwidth]{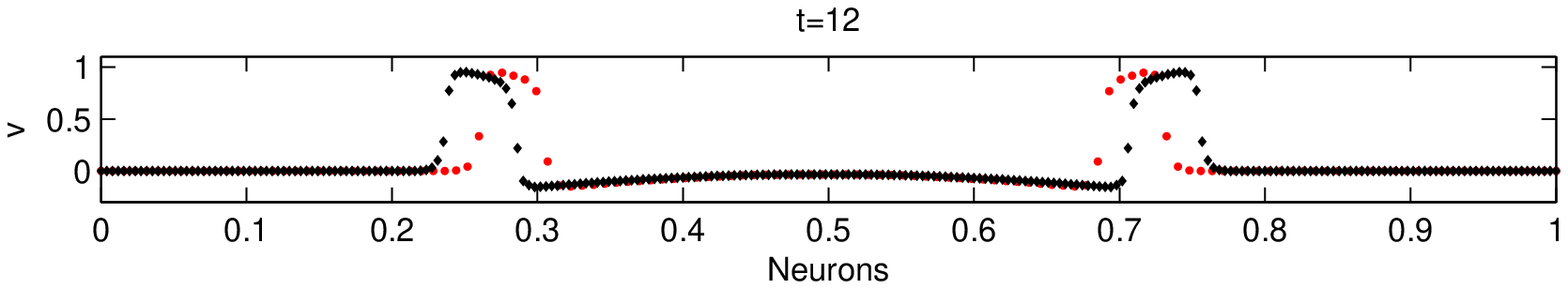}}
\end{minipage}
\begin{minipage}[]{\textwidth}
\centering
  \subfigure{\includegraphics[width=0.65\textwidth]{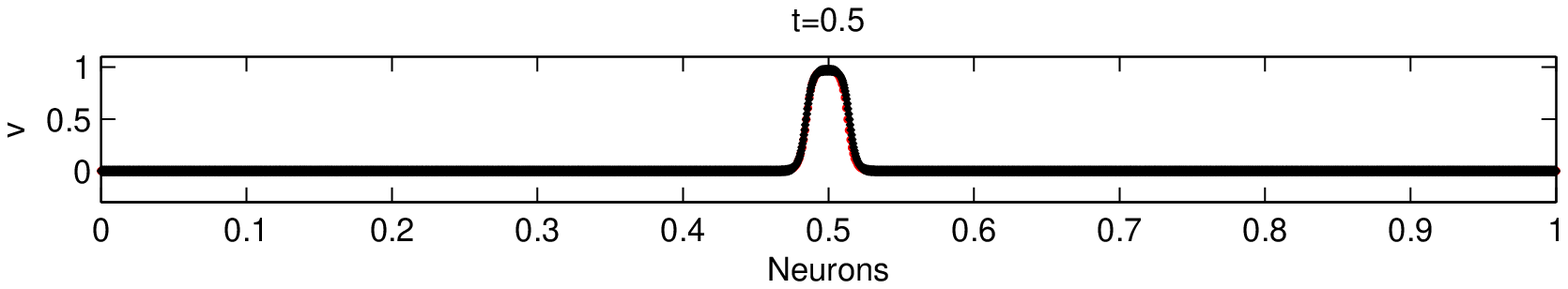}}\\\vspace{-0.65cm}
  \subfigure{\includegraphics[width=0.65\textwidth]{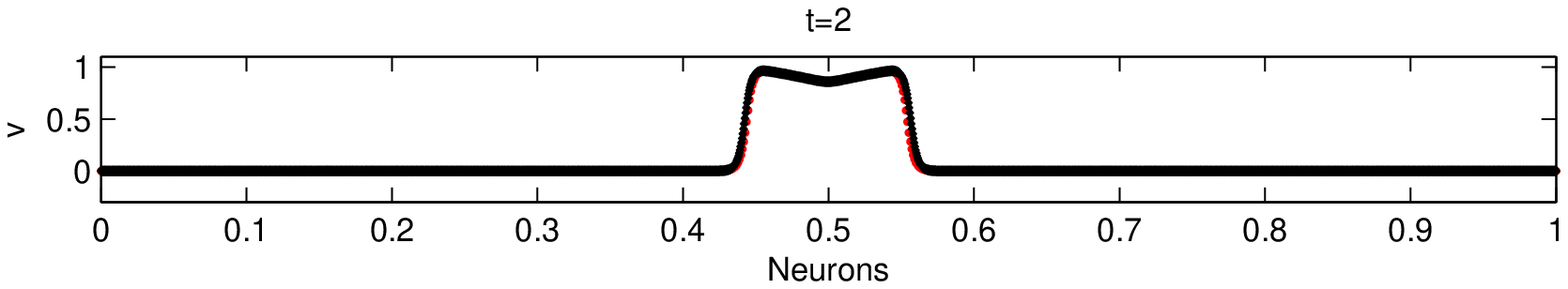}}\\\vspace{-0.65cm}
  \subfigure{\includegraphics[width=0.65\textwidth]{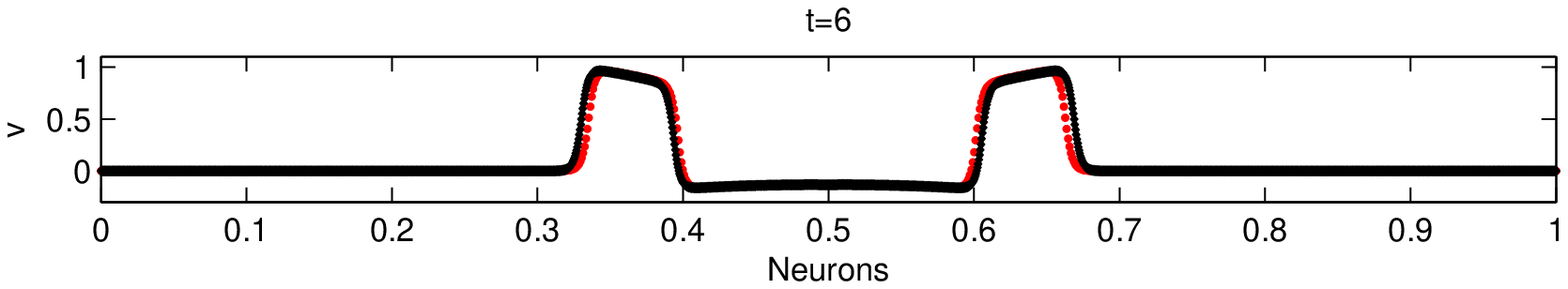}}\\\vspace{-0.65cm}
    \setcounter{subfigure}{1}
  \subfigure[]{\includegraphics[width=0.65\textwidth]{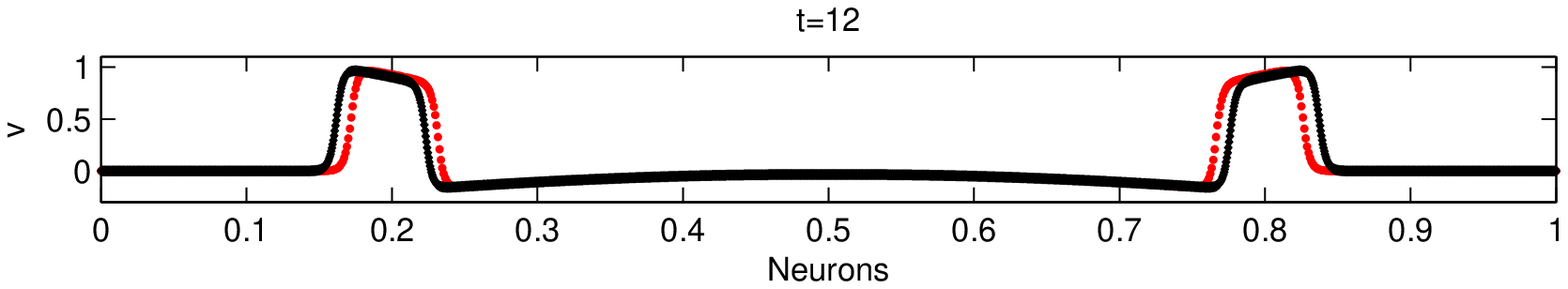}}
\end{minipage}
\caption{\label{ConfrS2} 
Convergence of the discrete model~\eqref{Eq:FHNiQ}-\eqref{eq:QN}-\eqref{Eq:FindQd} (Approach II) as $N \to \infty$. Evolution 
of pulses (a) for $N=128$ (red dots) and $N=256$ (black dots), (b) for $N=1024$ (red dots) and $N=2048$ (black dots)
}
\end{figure}

\begin{figure}
  \centering
  \subfigure{\includegraphics[width=0.65\textwidth]{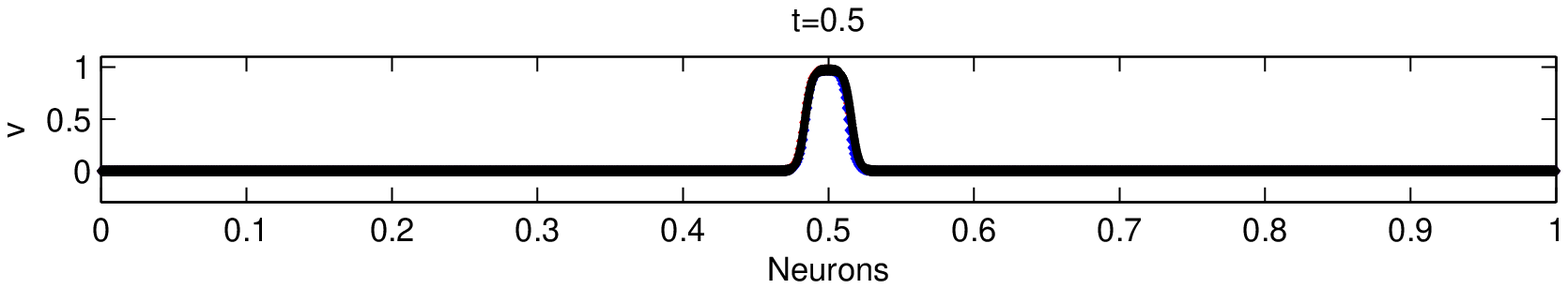}}\\\vspace{-0.65cm}
  \subfigure{\includegraphics[width=0.65\textwidth]{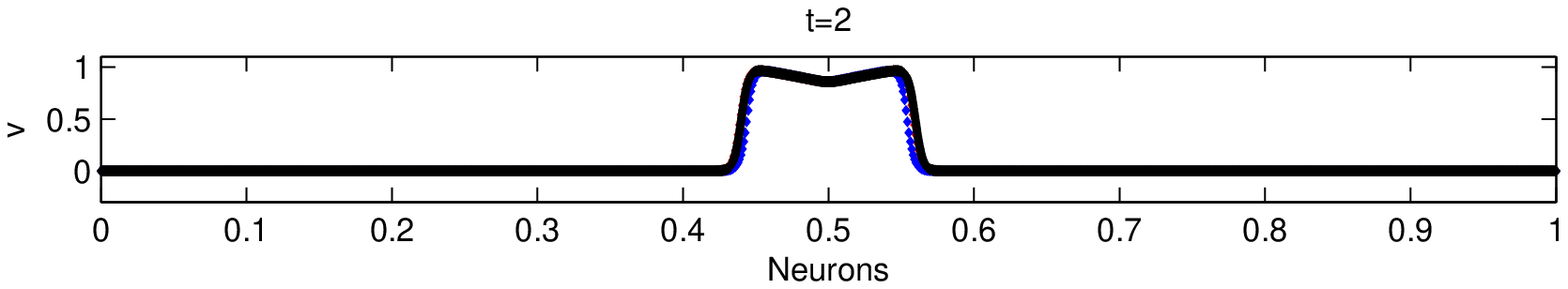}}\\\vspace{-0.65cm}
  \subfigure{\includegraphics[width=0.65\textwidth]{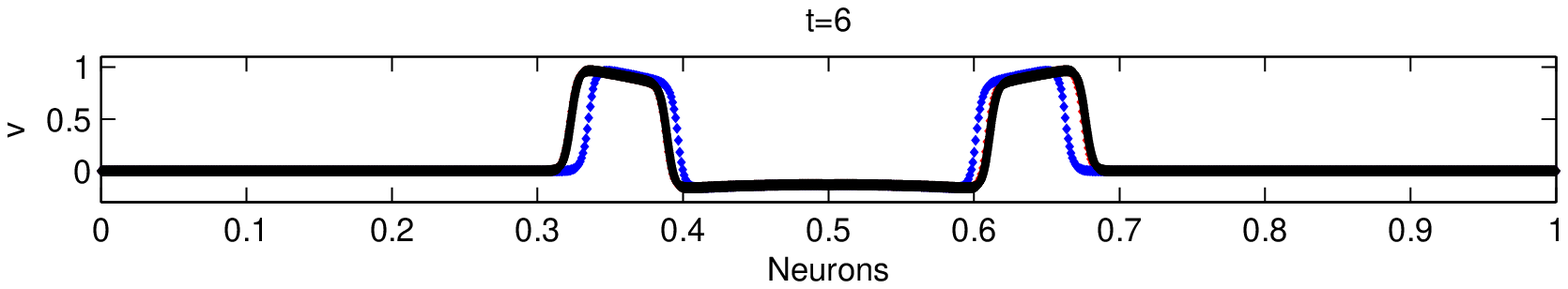}}\\\vspace{-0.65cm}
  \subfigure{\includegraphics[width=0.65\textwidth]{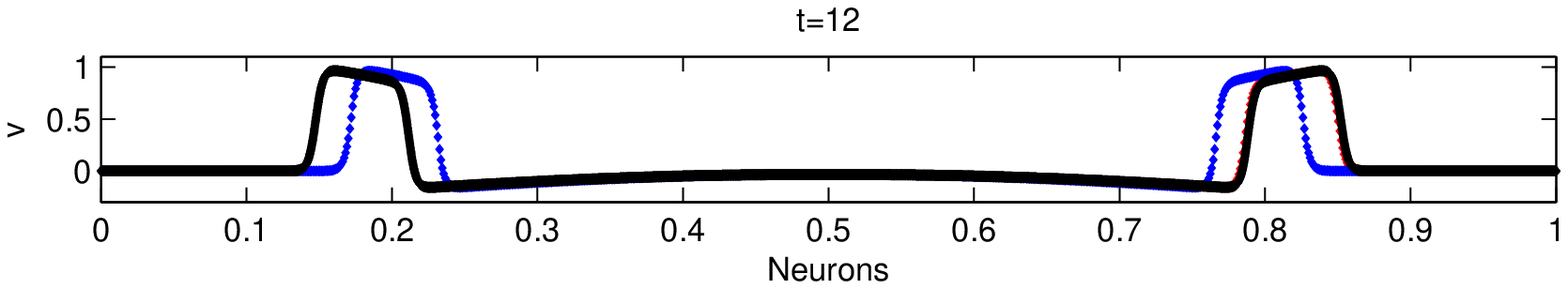}}
\caption{\label{ConfrS12_1024} Comparison of the dynamics produced by Approach II with $N=1024$ (blue dots)
and by the continuous model~\eqref{Eq:LimitSym} (black dots)
}
\end{figure}

\subsection{Non-symmetric interactions}
\label{Sec:NSI}
A more general configuration of the network admits non-symmetric links for each neuron, which correspond to unidirectional connections (the so-called rectifying synapses). A natural extension of the symmetric case consists in choosing 
\begin{equation}
\label{Eq:QAsym}
\mathcal{Q}=\mathcal{Q}_N=\mathcal{Q}_N^D\cup\mathcal{Q}_N^C\;,
\end{equation}
where
\begin{equation}
\nonumber
\mathcal{Q}_N^D=\{\pm 1,\,\cdots ,\,\pm Q_N^D\}\;,\quad
\mathcal{Q}_N^C=\{Q_N^D+1,\,\cdots ,\, Q_N^C\}\;.
\end{equation}
for some integers $Q_N^D\geq 1$ and $Q_N^C>Q_N^D$. (Choosing $-\mathcal{Q}_N^C$ instead of $\mathcal{Q}_N^C$ would be an obvious alternative.) We will prove that a suitable choice of $Q_N^C$ depending on $N$ leads to modify the limit model~\eqref{Eq:LimitSym}, by adding a first order term to the action potential equation.

With our definitions, the sum in the diffusive coupling becomes
\begin{equation}
\label{Eq:DiffCoupAsym}
\sum_{q\in\mathcal{Q}_N}(v_{i+q}-v_i)= \sum_{q=-Q_N^D}^{Q_N^D}(v_{i+q}-v_i) + \sum_{q=Q_N^D+1}^{Q_N^C}(v_{i+q}-v_i)\;.
\end{equation}
Exploiting the Taylor expansion~\eqref{Eq:Taylor}, we obtain
\begin{equation}
\begin{aligned}
&\sum_{q\in\mathcal{Q}_N}(v_{i+q}-v_i)= \\ 
&\left(\sum_{q=Q_N^D+1}^{Q_N^C}q \right)\Delta x v^\prime (x_i)+
\left(\sum_{q=1}^{Q_N^D} q^2 + \frac{1}{2}\sum_{q=Q_N^D+1}^{Q_N^C}q^2\right)\Delta x^2 v^{\prime\prime}(x_i)  +\mbox{h.o.t.}\;.
\end{aligned}
\label{Eq:AsymInitial}
\end{equation}
Recalling the definition~\eqref{Eq:phi} of the function $\varphi$, and introducing the function $\psi :\reali^{+}\rightarrow\reali^{+}$ defined as 
\begin{equation}
\label{Eq:varphi}
\psi (x)=\frac{x(x+1)}{2}
\end{equation}
and such that $\sum_{q=1}^{n} q=\psi(n)$, it is easily seen that the diffusive coupling takes the form
\begin{equation}
\label{Eq:Asymcase}
\begin{aligned}
-(L_Gv)_i=d&\left[(\psi(Q_N^C)-\psi(Q_N^D))\Delta x v^{\prime}(x_i)\right.\\
&\left. +\frac{1}{2} (\varphi(Q_N^D) +\varphi(Q_N^C))\Delta x^2 v^{\prime\prime}(x_i)+\mbox{h.o.t}\right]\;.
\end{aligned}
\end{equation}

Ideally, we would like to find integers $Q_N^D$ and $Q_N^C>Q_N^D$ satisfying the system
\begin{equation}
\label{Eq:AsymCompleteFirst}
\left\{\begin{array}{l}
\displaystyle \frac{1}{2} d(\phi(Q_N^D)+\phi(Q_N^C))\frac{1}{N^2}=d^\ast\\[10pt]
\displaystyle d(\psi (Q_N^C)-\psi (Q_N^D))\frac{1}{N}=c^\ast\;,
\end{array}\right.
\end{equation}
for fixed constants $d^\ast ,c^\ast >0$. At first, we discuss the existence of real solutions $Q_N^{D,r}$ and $Q_N^{C,r}$.
\begin{proposition}
Set $A_N=2\frac{d^\ast}{d}N^2$ and $B_N=\frac{c^\ast}{d}N$. If
\begin{equation}
\label{Eq:phipsi}
\varphi(\psi ^{-1}(B_N))\leq A_N\;,
\end{equation}
there exists a unique solution $(Q_N^{D,r},Q_N^{C,r})\in\reali_+^2$ of the previous system.
\end{proposition}
\begin{proof}
For simplicity, let us set $\hat{x}=Q_N^{D,r}$ and $\hat{y}=Q_N^{C,r}$. They should satisfy
\begin{equation}
\label{Eq:Cond}
\left\{
\begin{array}{l}
\displaystyle\varphi(\hat{x})+\varphi(\hat{y})=A_N\\[10pt]
\displaystyle\psi(\hat{y})-\psi(\hat{x})=B_N\;.
\end{array}
\right.
\end{equation}
Recalling that both $\varphi$ and $\psi$ are strictly increasing bijections from $[0,+\infty )$ into itself, the second equation yields
\begin{displaymath}
\hat{y}=\psi ^{-1}(\psi(\hat{x})+B_N)\;,
\end{displaymath}
which, substituted into the first equation, yields
\begin{equation}
\label{Eq:psiics}
\varphi(\hat{x})+\varphi(\eta(\hat{x}))=A_N\;,
\end{equation}
with $\eta(\hat{x}):=\psi^{-1}(\psi(\hat{x})+B_N)$. Now the function $\chi=\varphi+\varphi\circ\eta$ is again strictly increasing, and maps $[0,+\infty)$ into $[\chi(0),+\infty)=[\varphi(\psi^{-1}(B_N)),+\infty)$. Thus, condition~\eqref{Eq:phipsi} is equivalent to the existence of a unique solution of~\eqref{Eq:psiics}, whence the result.
\end{proof}
We observe that, given any arbitrary $d^\ast$ and $c^\ast$, there always exists an integer $N^\ast$ such that condition~\eqref{Eq:phipsi} is satisfied for all $N\geq N^\ast$.

\begin{definition}
\label{Def:QdQc}
Under the assumption~\eqref{Eq:phipsi}, we define $Q_N^{D}$ and $Q_N^{C}$, resp., to be the nearest integers to $Q_N^{D,r}$ and $Q_N^{C,r}$, resp., which are the unique solutions of the system
\begin{equation}
\label{Eq:AsymComplete}
\left\{\begin{array}{l}
\displaystyle \frac{1}{2} d(\phi(Q_N^{D,r})+\phi(Q_N^{C,r}))\frac{1}{N^2}=d^\ast\\[10pt]
\displaystyle d(\psi (Q_N^{C,r})-\psi (Q_N^{D,r}))\frac{1}{N}=c^\ast\;.
\end{array}\right.
\end{equation}
\end{definition}

\begin{proposition}
\label{Prop:QdQcLaw}
The following asymptotic behaviour of the integers $Q_N^D$ and $Q_N^C$ holds:
\begin{equation}
\nonumber
Q_N^D\simeq N^\frac{2}{3},\qquad Q_N^C\simeq N^\frac{2}{3}\quad\mbox{with}\quad Q_N^C-Q_N^D\simeq N^\frac{1}{3}.
\end{equation}
\end{proposition}
\begin{proof}
It is enough to estimate $\hat{x}=Q_N^{D,r}$ and $\hat{y}=Q_N^{C,r}$. We recall that they satisfy~\eqref{Eq:Cond}. With the ansatz $\hat{y}\simeq N^\alpha$, we have $\varphi(\hat{y})\simeq N^{3\alpha}$. On the other hand, the inequality $\hat{x}<\hat{y}$ and the monotonicity of $\varphi$ yield $\varphi(\hat{y})\leq \varphi(\hat{x})+\varphi(\hat{y})\leq 2\varphi(\hat{y})$. Since $A_N\simeq N^2$, we deduce that $\varphi(\hat{y})\simeq N^2$, whence $\alpha=2/3$. On the other hand, $\psi(\hat{y})\simeq N^\frac{4}{3}$ so that $\psi(\hat{x})=\psi(\hat{y})+B_N\simeq N^\frac{4}{3}+N\simeq N^\frac{4}{3}$, which implies $\hat{x}\simeq N^\frac{2}{3}$. Finally, by Lagrange's theorem, 
\begin{equation}
\nonumber
N\simeq B_N=\psi(\hat{y})-\psi(\hat{x})=\psi^{\prime}(\hat{z})(\hat{y}-\hat{x})
\end{equation}
for some $\hat{x}<\hat{z}<\hat{y}$; since $\psi^\prime (\hat{z})=\hat{z}+1/2\simeq N^\frac{2}{3}$, we conclude that $\hat{y}-\hat{x}\simeq N^\frac{1}{3}$.
\end{proof}
Even for the present model, interactions are local. Indeed, all neurons linked to the $i-$th one belong to the interval
\begin{equation}
\nonumber
I=[x_i-Q_N^D\Delta x,x_i+Q_N^C\Delta x]\;,
\end{equation}
whose length shrinks to $0$ as $N\rightarrow\infty$ since $Q_N^D\Delta x,\,Q_N^C\Delta x\simeq N^{-\frac{1}{3}}$.

In order to accommodate the effect of the slight shift from $(Q_N^{D,r},Q_N^{C,r})$ to $(Q_N^{D},Q_N^{C})$, we introduce perturbations $(d_N^\ast,c_N^\ast)$ of $(d^\ast,c^\ast)$. They are defined in such a way that $(Q_N^{D},Q_N^{C})$ is the solution of the system
\begin{equation}
\label{Eq:AsymCompleteStar}
\left\{\begin{array}{l}
\displaystyle \frac{1}{2} d(\phi(Q_N^D)+\phi(Q_N^C))\frac{1}{N^2}=d_N^\ast\\[10pt]
\displaystyle d(\psi (Q_N^C)-\psi (Q_N^D))\frac{1}{N}=c_N^\ast\;.
\end{array}\right.
\end{equation}
The size of the perturbation can be estimated as follows.
\begin{proposition}
\label{Prop:ddastccast}
The perturbed coefficients $d_N^\ast$ and $c_N^\ast$ introduced above satisfy
\begin{equation}
\nonumber
|d_N^\ast -d^\ast|\lesssim N^{-\frac{2}{3}},\quad |c_N^\ast -c^\ast|\lesssim N^{-\frac{1}{3}}\;.
\end{equation}
\end{proposition}
\begin{proof}
Using~\eqref{Eq:AsymComplete} and~\eqref{Eq:AsymCompleteStar}, we get
\begin{equation}
\nonumber
\begin{aligned}
d_N^\ast -d^\ast &= \frac{d}{2N^2}\left[(\varphi (Q_N^D)-\varphi (Q_N^{D,r}))+(\varphi (Q_N^C)-\varphi (Q_N^{C,r})) \right]\;,\\
c_N^\ast -c^\ast &= \frac{d}{N}\left[(\psi (Q_N^D)-\psi(Q_N^{D,r}))-(\psi(Q_N^C)-\psi(Q_N^{C,r})) \right]\;.
\end{aligned}
\end{equation}
As in the proof of Proposition~\ref{PropdNd}, we have $|\varphi (Q_N^D)-\varphi (Q_N^{D,r})|\lesssim N^{\frac{4}{3}}$, $|\varphi (Q_N^C)-\varphi (Q_N^{C,r})|\lesssim N^{\frac{4}{3}}$, and $|\psi (Q_N^D)-\psi (Q_N^{D,r})|\lesssim N^{\frac{2}{3}}$, $|\psi (Q_N^C)-\psi (Q_N^{C,r})|\lesssim N^{\frac{2}{3}}$, which gives the result.
\end{proof}

Finally, we study the limit behaviour of our model as $N\rightarrow\infty$. To this end, we make use of the following expression for the higher order terms in~\eqref{Eq:Asymcase}:
\begin{equation}
\mbox{h.o.t.}=\frac{1}{12}\sum_{q=1}^{Q_N^D}q^4\Delta x^4 v^{(iv)}(\bar{x}_{i,q})+\frac{1}{6}\sum_{q=Q_N^D+1}^{Q_N^C}q^3 \Delta x^3 v^{\prime\prime\prime}(\bar{\bar{x}}_{i,q})\;,
\end{equation}
which holds under the assumption that the fourth derivative of $v$ is continuous in $[0,1]$, for suitable points $\bar{x}_{i,q}\in (x_{i-q},x_{i+q})$ and $\bar{\bar{x}}_{i,q}\in (x_{i},x_{i+q})$.  Then, we observe that
\begin{equation}
\nonumber
\sum_{q=1}^{Q_N^D}q^4\Delta x^4\simeq (Q_N^D)^5\Delta x^4\simeq N^{\frac{10}{3}-4}=N^{-\frac{2}{3}}
\end{equation}
and
\begin{equation}
\nonumber
\sum_{q=Q_N^D+1}^{Q_N^C}q^3 \Delta x^3 \simeq \left[(Q_N^C)^4-(Q_N^D)^4 \right]\Delta x^3\simeq N^{\frac{7}{3}-3}=N^{-\frac{2}{3}}\;.
\end{equation}
Thus, we obtain the following result.
\begin{theorem}
Fix any point $\hat{x}\in [0,1]$ and for each $N$, consider a neuron $i=i(N)$ such that $x_{i(N)}\rightarrow \hat{x}$ as $N\rightarrow\infty$. Assuming the continuity of the fourth derivative of $v$ in $[0,1]$, we have
\begin{equation}
\nonumber
-(L_Gv)_{i(N)}\rightarrow d^\ast v''(\hat{x})+c^\ast v'(\hat{x})\quad\mbox{as }N\rightarrow\infty\;.
\end{equation}
Therefore, the discrete model~\eqref{Eq:FHNiQ} with $\mathcal{Q}$ given by~\eqref{Eq:QAsym} and $Q_N^D,\,Q_N^C$ defined in Definition~\ref{Def:QdQc} leads for $N\rightarrow\infty$ to the continuous model 
\begin{equation}
\begin{array}{l}
\displaystyle\frac{\partial v}{\partial t}= f(v,r)+d^\ast\frac{\partial ^2 v}{\partial x^2}+c^\ast\frac{\partial v}{\partial x}\;,\\[10pt]
\displaystyle\frac{\partial r}{\partial t}=g(v,r)\;,
\end{array}
\label{Eq:LimitEqAsym}
\end{equation}
which describes the behaviour of a continuum of neurons disposed along a closed ring.
\begin{flushright}
$\square$
\end{flushright}
\end{theorem}

\begin{remark}\label{rem:comments1}
A few comments are in order.
\begin{itemize}
\item[i)] Observe that having a larger number of neurons influencing a given neuron from its right rather than from
its left results in a convective term, whose coefficient $c^\ast$ is positive; this corresponds to a negative speed of
convective propagation, i.e., waves moving from right to left, as documented by Fig.~\ref{ConfrAL}. Obviously,
choosing $c^\ast=0$ yields ${\mathcal Q}_N^C=\emptyset$, so one is back to the symmetric case considered in Sect. ~\ref{Sec:Extended1D}.
\item[ii)] The same limit model can be obtained with a nearest-neighbour interaction that extends the one considered
in Sect.~\ref{NNI}, i.e.,  
\begin{equation}
-(L_Gv)_i=d_N[(v_{i+1}-v_i)+(v_{i-1}-v_i)] + c_N (v_{i+2}-v_i)\;,
\label{Eq:DiffCoupMod}
\end{equation}
with $d_N=d^\ast N^2$ and $c_N=c^\ast \frac{N}2$.
\item[iii)] A generalization to variable coefficients $d^\ast$ and $c^\ast$ similar to the one discussed in 
Remark~\ref{rem:var-coeff} is also possible, yielding the two last terms on the right-hand side of~\eqref{Eq:LimitEqAsym}
being replaced by the conservation form $\frac{\partial}{\partial x}\left(d^\ast\frac{\partial v}{\partial x} \right)+ \frac{\partial}{\partial x}\left(c^\ast v \right)$.
\end{itemize}
\end{remark}

We now provide some quantitative insights for our model. Extending the test case considered in the previous subsection, we choose $d=0.05$ and we enforce that for $N=N_0=128$, we have $Q_{N_0}^D=Q_{N_0}^{D,r}=1$ and $Q_{N_0}^C=Q_{N_0}^{C,r}=2$, i.e., each neuron is influenced by its first neighbour on the left and by the two first neighbours on the right. Using~\eqref{Eq:AsymComplete}, we obtain
\begin{equation}
\begin{aligned}
d^\ast &= \frac{3\cdot 0.05}{128^2}=9.1553\cdot 10^{-6}\;,\\
c^\ast &= \frac{2\cdot 0.05}{128}=7.8125\cdot 10^{-4}\;.
\end{aligned}
\end{equation}
Then, we increase $N$ by powers of $2$ and we monitor the evolution of the quantities $Q_N^D$ and $Q_N^C$, as well as the errors $d_N^\ast -d^\ast$ and $c_N^\ast -c^\ast$. The results, reported in Table~\ref{Tab:QDQC}, indicate an excellent agreement with the theoretical predictions given in Propositions~\ref{Prop:QdQcLaw}--\ref{Prop:ddastccast}. The evolutions of the action potentials produced by the discrete model with $N=512$ and $N=2048$, and by a very accurate solution of the continuous model~\eqref{Eq:LimitEqAsym} are documented in Figure~\ref{ConfrAL}.

\begin{figure}
  \centering
  \subfigure{\includegraphics[width=0.65\textwidth]{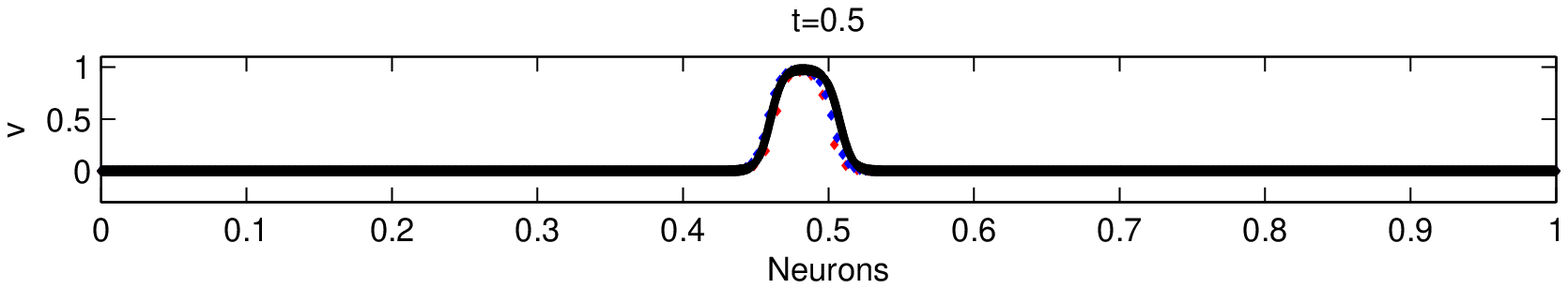}}\\\vspace{-0.65cm}
  \subfigure{\includegraphics[width=0.65\textwidth]{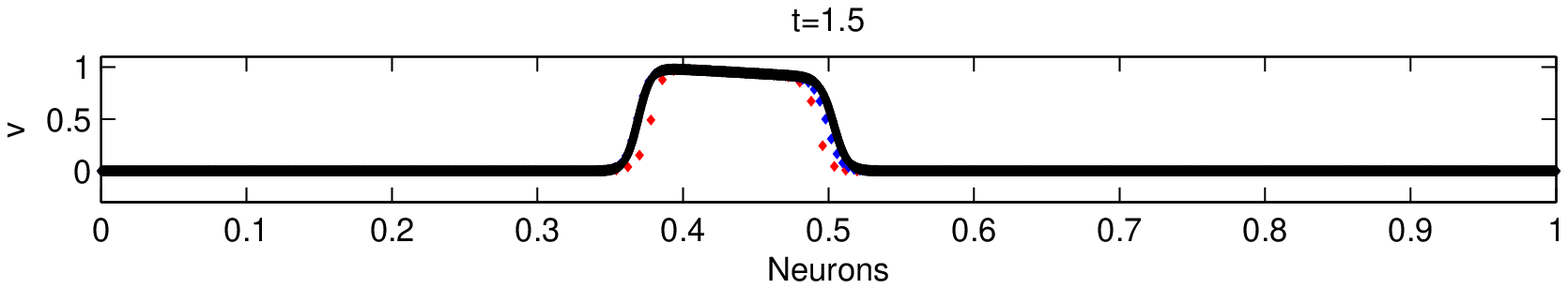}}\\\vspace{-0.65cm}
  \subfigure{\includegraphics[width=0.65\textwidth]{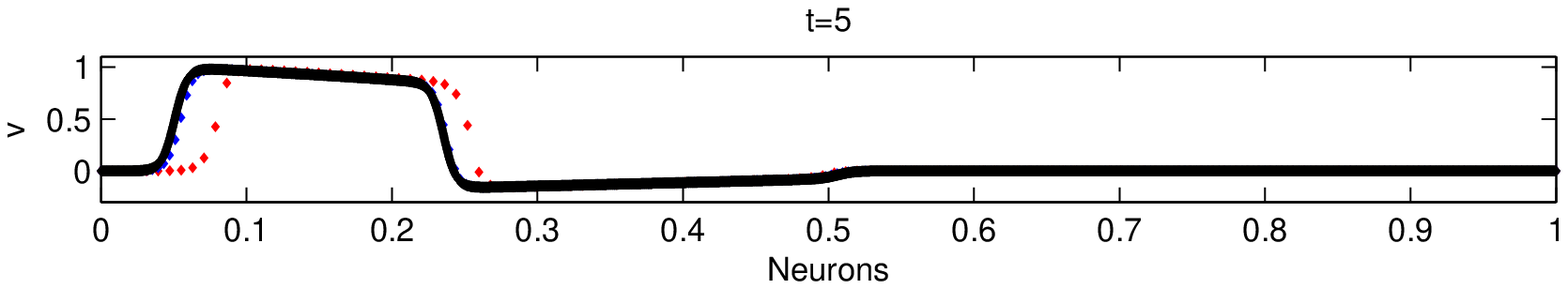}}\\\vspace{-0.65cm}
  \subfigure{\includegraphics[width=0.65\textwidth]{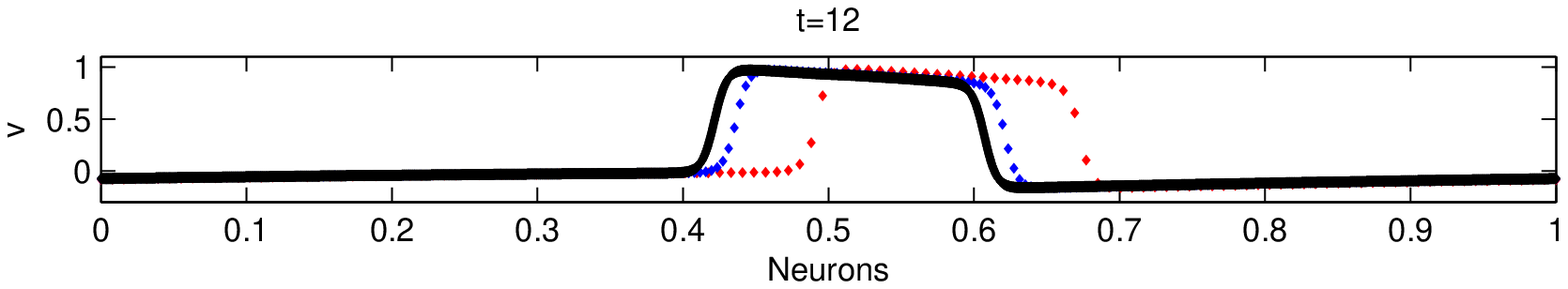}}
\caption{\label{ConfrAL}
Convergence of the discrete model~\eqref{Eq:FHNiQ}-\eqref{Eq:QAsym}-\eqref{Eq:AsymComplete} as $N \to \infty$. Evolution 
of a pulse  for $N=128$ (red dots), $N=256$ (blue dots) and $N=8192$ (black dots)}
\end{figure}

\begin{table}[t!]\footnotesize
\centering
\caption{Number of connections per neuron $Q_D$, $Q_C$, convection coefficient $\tilde{c}_N$ and diffusion coefficient $\tilde{d}_N$ as a function of $N$ are shown.}
\begin{tabular}{rrrrlclc}
$p$ & $N=N=2^p$& $Q_N^D$ & $Q_N^C$ & $N_0 d_N^\ast$ & $|d_N^\ast-d_N|/d_N$ & $N_0 c_N^\ast$ & $|c_N^\ast-c_N|/c_N$\\
\hline
$0$ & $128$ \rule{0pt}{2.5ex} \rule[-1ex]{0pt}{0pt} & 1 & 2 & 0.1500 & 0& 0.1000 &0 \\
\hline
$1$ & $256$ \rule{0pt}{2.5ex} \rule[-1ex]{0pt}{0pt} & 2 & 3 & 0.1188 & 2.1$\cdot 10^{-1}$ & 0.0750 &  2.5$\cdot 10^{-1}$\\
\hline
$2$ & $512$ \rule{0pt}{2.5ex} \rule[-1ex]{0pt}{0pt} & 4 & 5 & 0.1328 & 1.1$\cdot 10^{-1}$ & 0.0625 & 3.75$\cdot 10^{-1}$ \\
\hline
$3$ & $1024$ \rule{0pt}{2.5ex} \rule[-1ex]{0pt}{0pt} & 7 & 9 & 0.1660 & 1.1$\cdot 10^{-1}$ & 0.1063 & 6.25$\cdot 10^{-2}$ \\
\hline
$4$ & $2048$ \rule{0pt}{2.5ex} \rule[-1ex]{0pt}{0pt} & 11 & 14 & 0.1485 & 9.8$\cdot 10^{-3}$ & 0.1219 & 2.2$\cdot 10^{-1}$\\
\hline
$5$ & $4096$ \rule{0pt}{2.5ex} \rule[-1ex]{0pt}{0pt} & 19 & 22 & 0.1524 & 2.0$\cdot 10^{-2}$ & 0.0984 & 1.6$\cdot 10^{-2}$\\
\hline
$6$ & $8192$ \rule{0pt}{2.5ex} \rule[-1ex]{0pt}{0pt} & 31 & 35 & 0.1546  & 3.0$\cdot 10^{-2}$ &  0.1047 & 4.7$\cdot 10^{-2}$\\
\hline
$7$ & $16384$ \rule{0pt}{2.5ex} \rule[-1ex]{0pt}{0pt} & 50 & 55 & 0.1524 & 1.6$\cdot 10^{-2}$ & 0.1035 & 3.5$\cdot 10^{-2}$ \\
\hline
$8$ & $32768$ \rule{0pt}{2.5ex} \rule[-1ex]{0pt}{0pt} & 80 & 86 & 0.1486 & 9.2$\cdot 10^{-3}$ & 0.0979 & 2.1$\cdot 10^{-2}$ \\
\hline
$9$ & $65536$ \rule{0pt}{2.5ex} \rule[-1ex]{0pt}{0pt} & 129 & 136  & 0.1499 & 7.6$\cdot 10^{-4}$ & 0.0909 & 9.1$\cdot 10^{-2}$\\
\hline
$10$ & $131072$ \rule{0pt}{2.5ex} \rule[-1ex]{0pt}{0pt} & 206 & 216 & 0.1506 & 4.2$\cdot 10^{-3}$ & 0.1033 & 3.3$\cdot 10^{-2}$\\
\hline
$11$ & $262144$ \rule{0pt}{2.5ex} \rule[-1ex]{0pt}{0pt} & 329 & 341 & 0.1502 & 1.4$\cdot 10^{-3}$ & 0.0983 & 1.7$\cdot 10^{-2}$\\
\hline
$12$ & $524288$ \rule{0pt}{2.5ex} \rule[-1ex]{0pt}{0pt} & 524 &  540 & 0.1501 & 6.7$\cdot 10^{-4}$ & 0.1040 & 4.0$\cdot 10^{-2}$\\
\hline
$13$ & $1048576$ \rule{0pt}{2.5ex} \rule[-1ex]{0pt}{0pt} & 835 & 854 & 0.1499 & 6.6$\cdot 10^{-4}$ & 0.0980 & 2.0$\cdot 10^{-2}$\\
\hline
$14$ & $2097152$ \rule{0pt}{2.5ex} \rule[-1ex]{0pt}{0pt} & 1329 & 1353 & 0.1499 & 4.7$\cdot 10^{-4}$ & 0.0982 & 1.7$\cdot 10^{-2}$\\
\hline
$15$ & $4194304$ \rule{0pt}{2.5ex} \rule[-1ex]{0pt}{0pt} & 2114 & 2145 & 0.1500 & 1.5$\cdot 10^{-4}$ & 0.1008 & 7.5$\cdot 10^{-3}$\\
\hline
$16$ & $8388608$ \rule{0pt}{2.5ex} \rule[-1ex]{0pt}{0pt} & 3361 & 3400 & 0.1499 & 5.4$\cdot 10^{-5}$ & 0.1006 & 6.0$\cdot 10^{-2}$\\
\hline
$17$ & $16777216$ \rule{0pt}{2.5ex} \rule[-1ex]{0pt}{0pt} & 5342 & 5391 & 0.1499 & 9.3$\cdot 10^{-5}$ & 0.1003 & 3.2$\cdot 10^{-3}$\\
\hline
$18$ & $33554432$ \rule{0pt}{2.5ex} \rule[-1ex]{0pt}{0pt} & 8489 & 8550 & 0.1500 & 3.0$\cdot 10^{-5}$  & 0.0991 & 8.7$\cdot 10^{-3}$\\
\hline
$19$ & $67108864$ \rule{0pt}{2.5ex} \rule[-1ex]{0pt}{0pt} &  13485 & 13563 & 0.1500 & 1.8$\cdot 10^{-5}$ & 0.1006 & 6.0$\cdot 10^{-3}$\\
\hline
$20$ & $134217728$ \rule{0pt}{2.5ex} \rule[-1ex]{0pt}{0pt} & 21420 & 21517 & 0.1500 & 5.5$\cdot 10^{-7}$ &0.0993 & 7.0$\cdot 10^{-3}$ \\
\end{tabular}
\label{Tab:QDQC}
\end{table}

\section{Multi-dimensional dynamics}
In this section, we extend the previous one-dimensional treatment, and in particular the material of Section~\ref{Sec:NSI}, to describe the dynamics of a multi-dimensional agglomeration of neurons. We will focus on the main aspects of the analysis, leaving to the reader those details that are straightforward extensions of the one-dimensional results.

We assume that neurons form a periodic lattice contained in $B=[0,1]^m$, $m=2$ or $m=3$. Precisely, given any integer $n\geq 2$ and setting $h=1/n$, each neuron is associated to a multi-index $l\in\{0,\cdots ,n-1\}^m$, which identifies its physical position $x=hl\in B$. Thus we have $N=n^m$ distinct neurons in $B$, which are labelled by indices $i\in\{1,\cdots ,N\}$ according to some rule; the $i-$th neurons has position $x_i=hl_i$, action potential $v_i$ and recovery variable $r_i$. Periodicity means that we replicate the situation at $x=hl$ in any $y=h(l+nk)$ with $k\in\interi ^m$.

We adopt again the diffusion model~\eqref{Eq:FHNiQ}, with $\mathcal{Q}$ given by~\eqref{Eq:QAsym}. The definition of $\mathcal{Q}_N^D$ and $\mathcal{Q}_N^C$ is as follows:
\begin{itemize}
\item given a radius $R_N^D:=hQ_N^D$ with $Q_N^D>0$ (to be determined later on), we set 
\begin{equation}
\label{eq:def-QND}
\mathcal{Q}_N^D :=\{q:\| x_{i+q}-x_i\|\leq R_N^D\}\;;
\end{equation}
\item given a radius $R_N^C:=hQ_N^C$ with $Q_N^C \geq Q_N^D$ (to be determined later on), and a unit vector $\nu\in\reali^m$, we set 
\begin{equation}
\label{eq:def-QNC}
\mathcal{Q}_N^C:=\{q:R_N^D<\| x_{i+q}-x_i\|\leq R_N^C\mbox{ and }(x_{i+q}-x_i)\cdot\nu\geq 0\}\;,
\end{equation}
i.e., $\mathcal{Q}_N^C$ identifies neurons sitting on semi-balls of suitable radii centered at $x_i$; these semi-balls are obtained by cutting the corresponding balls by the hyperplane containing $x_i$ and perpendicular to $\nu$, and retaining the halves oriented in the direction of $\nu$ (see Figure~\ref{Fig:Conf2D} for a  pictorial representation of the sets $\mathcal{Q}_N^D$ and $\mathcal{Q}_N^C$ in two dimensions).
\begin{figure}
\centering
\includegraphics[width=0.4\textwidth]{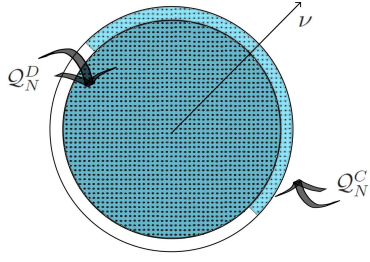}
\caption{\label{Fig:Conf2D} The sets $\mathcal{Q}_N^D$ and $\mathcal{Q}_N^C$  represented in a two-dimensional lattice}
\end{figure}
\end{itemize}
\subsection*{ The effect of $\mathcal{Q}_N^D$ on the diffusion term $-(L_Gv)_i$}
Observe that $q\in\mathcal{Q}_N^D$ iff $x_{i+q}=hl_{i+q}=h(l_i+k)$ for some $k\in\mathcal{K}_N^D:=\{k\in\interi ^m : \|k\|\leq Q_N^D\}$.
Thus, recalling~\eqref{Eq:TaylorGen}, we have
\begin{equation}
\label{Eq:DiffCoupQd}
\sum_{q\in\mathcal{Q}_N^D} (v_{i+q}-v_i)=\sum_{k\in\mathcal{K}_N^D}h \, k \cdot \! \nabla v(x_i)+\frac{1}{2}h^2k^T Hv(x_i)k+\mbox{h.o.t.}\;.
\end{equation}
Now, writing
\begin{equation}
\nonumber
k^THv(x_i)k=\sum_{\alpha =1}^m k_{\alpha}^2 D_{\alpha\alpha}^2 v(x_i)+\sum_{\substack{\alpha ,\beta=1\\ \alpha\neq\beta}}^m k_\alpha k_\beta D_{\alpha\beta}^2v(x_i)\;,
\end{equation}
we get
\begin{equation}
\nonumber
\sum_{k\in\mathcal{K}_N^D} k^THv(x_i)k=\sum_{\alpha =1}^m\left(\sum_{k\in\mathcal{K}_N^D} k_\alpha ^2\right)D_{\alpha\alpha}^2 v(x_i)+\sum_{\substack{\alpha ,\beta=1\\ \alpha\neq\beta}}^m \left(\sum_{k\in\mathcal{K}_N^D} k_\alpha k_\beta   \right) D_{\alpha\beta}^2v(x_i)\;.
\end{equation}
Now, it is easily seen that by the form of $\mathcal{K}_N^D$, the quantity
\begin{equation}
\nonumber
\varphi (Q_N^D):=\sum_{k\in\mathcal{K}_N^D} k_\alpha ^2,\qquad\mbox{with }\alpha =1,\cdots ,m
\end{equation}
is independent of $\alpha$, whereas
\begin{equation}
\sum_{k\in\mathcal{K}_N^D}k=0,\qquad \sum_{k\in\mathcal{K}_N^D} k_\alpha k_\beta=0
\quad \mbox{ if }\alpha\neq\beta\;,
\end{equation}
since vectors in $\mathcal{K}_N^D$ can be arranged in couples that are symmetric with respect to each coordinate hyperplane.
Thus, 
\begin{equation}
\label{Eq:sumQD}
\sum_{q\in\mathcal{Q}_N^D}(v_{i+q}-v_i)=\frac{1}{2}h^2\varphi(Q_N^D)\Delta v(x_i)+\mbox{h.o.t.}\;,
\end{equation}
where $\Delta v=\sum_{\alpha =1}^m D_{\alpha\alpha}^2v$ is the Laplacian of the function $v$. We observe for further reference that for any $Q>0$, denoting by $\mathcal{B}(0,Q)$ the ball of center $0$ and radius $Q$ in $\reali ^m$, one has for any given $\alpha =1,\cdots ,m$
\begin{equation}
\varphi(Q)=\sum_{\|k\|\leq Q}k_\alpha ^2\sim \int_{\mathcal{B}(0,Q)}y_\alpha ^2\d y\simeq Q^{2+m}\quad\mbox{as }Q\rightarrow\infty\;.
\end{equation}

\subsection*{The effect of $\mathcal{Q}_N^C$ on the diffusion term $-(L_Gv)_i$}
Now, $q\in\mathcal{Q}_N^C$ iff $x_{i+q}=h(l_i+k)$ for some 
$k\in\mathcal{K}_N^C:=\{k\in\interi ^m:Q_N^D< \|k\| \leq Q_N^C\mbox{ and }k\cdot \nu>0\}$.
At this point, we assume that $\nu =e_1$, the first element of the canonical basis in $\reali ^m$; this choice is not at all restrictive, but simplifies the following arguments. Indeed, referring to the analogue of~\eqref{Eq:DiffCoupQd} in which $\mathcal{Q}_N^D,\,\mathcal{K}_N^D$ resp., are replaced by $\mathcal{Q}_N^C,\,\mathcal{K}_N^C$ resp., we have
\begin{equation}
\nonumber
\sum_{k\in\mathcal{K}_N^C}k\cdot \! \nabla v(x_i)=\left(\sum_{k\in\mathcal{K}_N^C} k_1\right)\frac{\partial v}{\partial x_1}(x_i)=\left(\psi(Q_N^C)-\psi(Q_N^D)\right)\frac{\partial v}{\partial x_1}(x_i)\;,
\end{equation}
with
\begin{equation}
\nonumber
\psi(Q):=\sum_{\substack{\|k\|\leq Q\\k_1\geq 0}}k_1\sim \int_{\mathcal{B}(0,Q)\cap\{y_1\geq 0\}}y_1\d y\simeq Q^{1+m}\quad\mbox{as }Q\rightarrow\infty\;.
\end{equation}
On the other hand,
\begin{equation}
\nonumber
\sum_{k\in\mathcal{K}_N^C}k^THv(x_i)k=\sum_{\alpha =1}^m\left(\sum_{k\in\mathcal{K}_N^C} k_\alpha ^2 \right)D_{\alpha\alpha}^2v(x_i)\;.
\end{equation}
But now, 
\begin{equation}
\nonumber
\sum_{k\in\mathcal{K}_N^C}k_\alpha ^2=\frac{1}{2}\sum_{Q_N^D<\| k\|\leq Q_N^C}k_\alpha ^2=\frac{1}{2}\left(\varphi(Q_N^C)-\varphi(Q_N^D)\right)\;.
\end{equation}
We conclude that, going back to the case of an arbitrary $\nu$,
\begin{equation}
\label{Eq:sumQC}
\begin{aligned}
\sum_{q\in\mathcal{Q}_N^C}(v_{i+q}-v_i)&= d\left[h\left(\psi(Q_N^C)-\psi(Q_N^D)\right)\nu\cdot\nabla v(x_i)\right.\\
&\left. +\frac{1}{4}h^2\left(\varphi(Q_N^C)-\varphi(Q_N^D)\right)\Delta v(x_i)+\mbox{h.o.t.}\right]\;.
\end{aligned}
\end{equation}
\subsection*{The global effect of $\mathcal{Q}_N^C$}
Summing up~\eqref{Eq:sumQD} and~\eqref{Eq:sumQC}, we obtain
\begin{equation}
\nonumber
\begin{aligned}
-(L_Gv)_i=d&h\left(\psi(Q_N^C)-\psi(Q_N^D)\right)\nu\cdot \! \nabla v(x_i)\\
&+\frac{d}{4}h^2\left(\varphi(Q_N^D)+\varphi(Q_N^C)\right)\Delta v(x_i)+\mbox{h.o.t.}\;.
\end{aligned}
\end{equation}
At this point, given two constants $d^\ast >0$ and $c^\ast \geq 0$, we would like to find $Q_N^D>0$ and $Q_N^C\geq Q_N^D$ such that
\begin{equation}
\label{Eq:System2D}
\left\{\begin{array}{l}
\displaystyle d\frac{h^2}{4}(\phi(Q_N^D)+\phi(Q_N^C))=d^\ast\\[10pt]
\displaystyle dh(\psi (Q_N^C)-\psi (Q_N^D))=c^\ast\;.
\end{array}\right.
\end{equation}
This system is similar to~\eqref{Eq:AsymCompleteFirst} and we can discuss its solvability as done in Section~\ref{Sec:NSI}. The conclusion is that for $N$ large enough, the solution exists and satisfies 
\begin{equation}
\nonumber
Q_N^D\simeq Q_N^C\simeq N^{\frac{2}{m(m+2)}}\quad\mbox{and}\quad Q_N^C-Q_N^D\simeq c^\ast N^{\frac{2-m}{m(m+2)}}\;,
\end{equation}
whereas the number of neurons that should be connected to a given neuron scales like $N^\frac{2}{m+2}$.
We summarize our conclusions as follows.

\begin{theorem}\label{theo:multiD}
The discrete model~\eqref{Eq:FHNiQ}, with ${\mathcal Q}$ given by~\eqref{Eq:QAsym}-\eqref{eq:def-QND}-\eqref{eq:def-QNC}
in which $Q_N^D$ and $Q_N^C$ are the solution of~\eqref{Eq:System2D},
tends for $N\rightarrow\infty$ to the following continuous model of reaction-convection-diffusion type
\begin{equation}
\begin{array}{l}
\displaystyle\frac{\partial v}{\partial t}= f(v,r)+d^\ast\Delta v+\hat{c}^\ast\cdot\nabla v\;,\\[10pt]
\displaystyle\frac{\partial r}{\partial t}=g(v,r)\;,
\end{array}
\label{Eq:LimitEqAsym2D}
\end{equation}
where the convective velocity is given by the vector $\hat{c}^\ast=c^\ast\nu$. 
\end{theorem}

The well-posedness of this model, as well as its numerical discretization, can be studied by adapting the arguments given in~\cite{ColliSavare} and~\cite{Sanfelici}.

An example of a two-dimensional dynamics produced by the model described above is given in Figure~\ref{Fig:2D}. We fix $d=0.05$ as for the one-dimensional models; then, we choose $d^\ast$ and $c^\ast$ in such a way that~\eqref{Eq:System2D} is satisfied for $n=256$ by $Q_N^D=\sqrt{2}$ and $Q_N^C=2$. This gives
\begin{equation}
\nonumber
d^\ast=3.8147\cdot 10^{-6}\quad\mbox{and}\quad c^\ast=3.9063\cdot 10^{-4}\;.
\end{equation}
The vector $\nu$ is chosen to be $e_1$. Figure~\ref{Fig:2D} shows the evolution of the action potential in the periodic box $B=[0,1]^2$ for $n=256$, starting from an initial stimulus $v_{|t=0}=1$ applied to the neurons lying in the circle of radius $1/32$ around the center of the box. The stimulus propagates in all directions, but since $c^\ast>0$ the speed of propagation is 
faster in the direction of $-\nu$.

\begin{figure}[t]
\centering
  \subfigure{\includegraphics[width=0.48\textwidth]{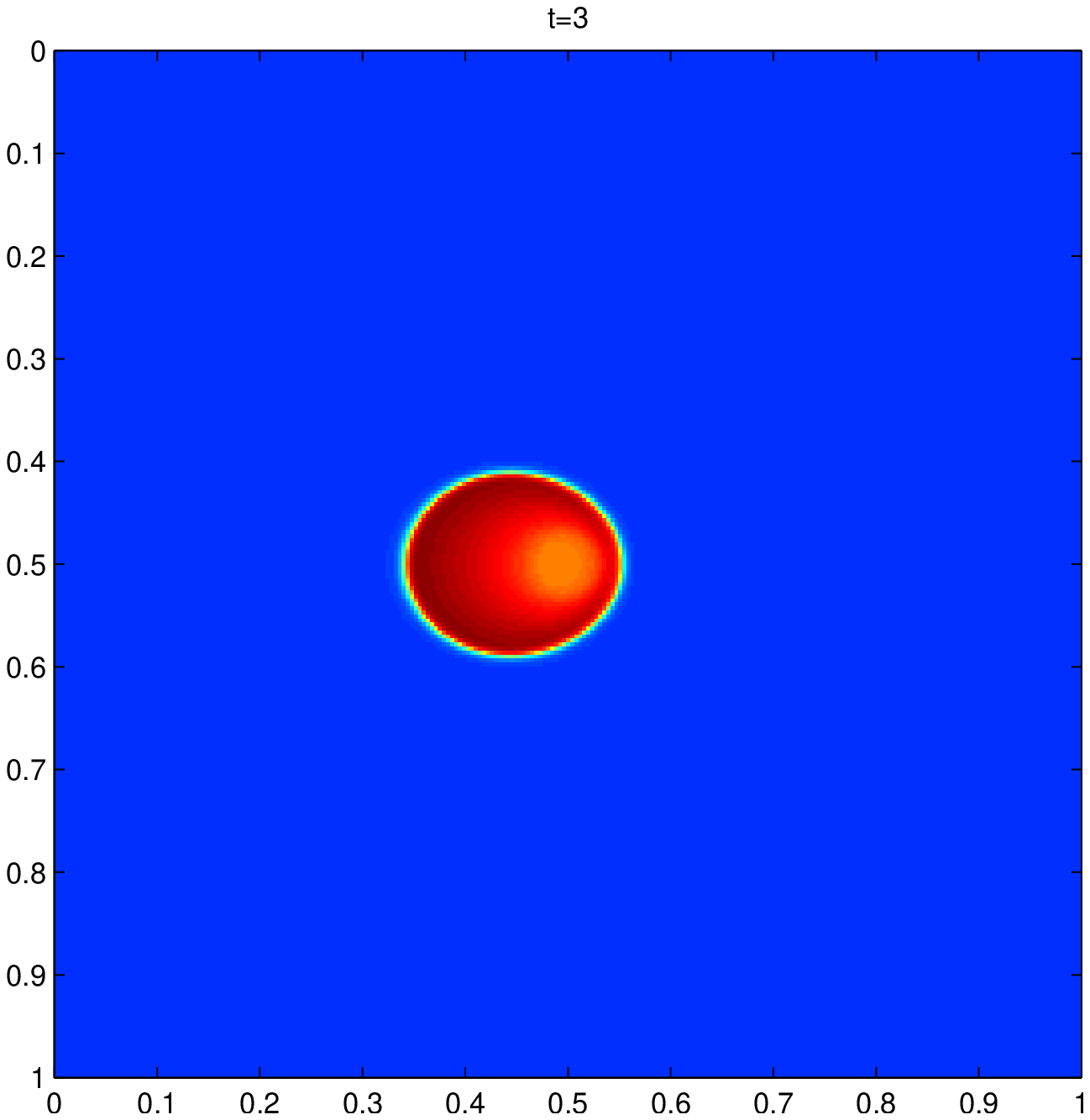}}
  \subfigure{\includegraphics[width=0.48\textwidth]{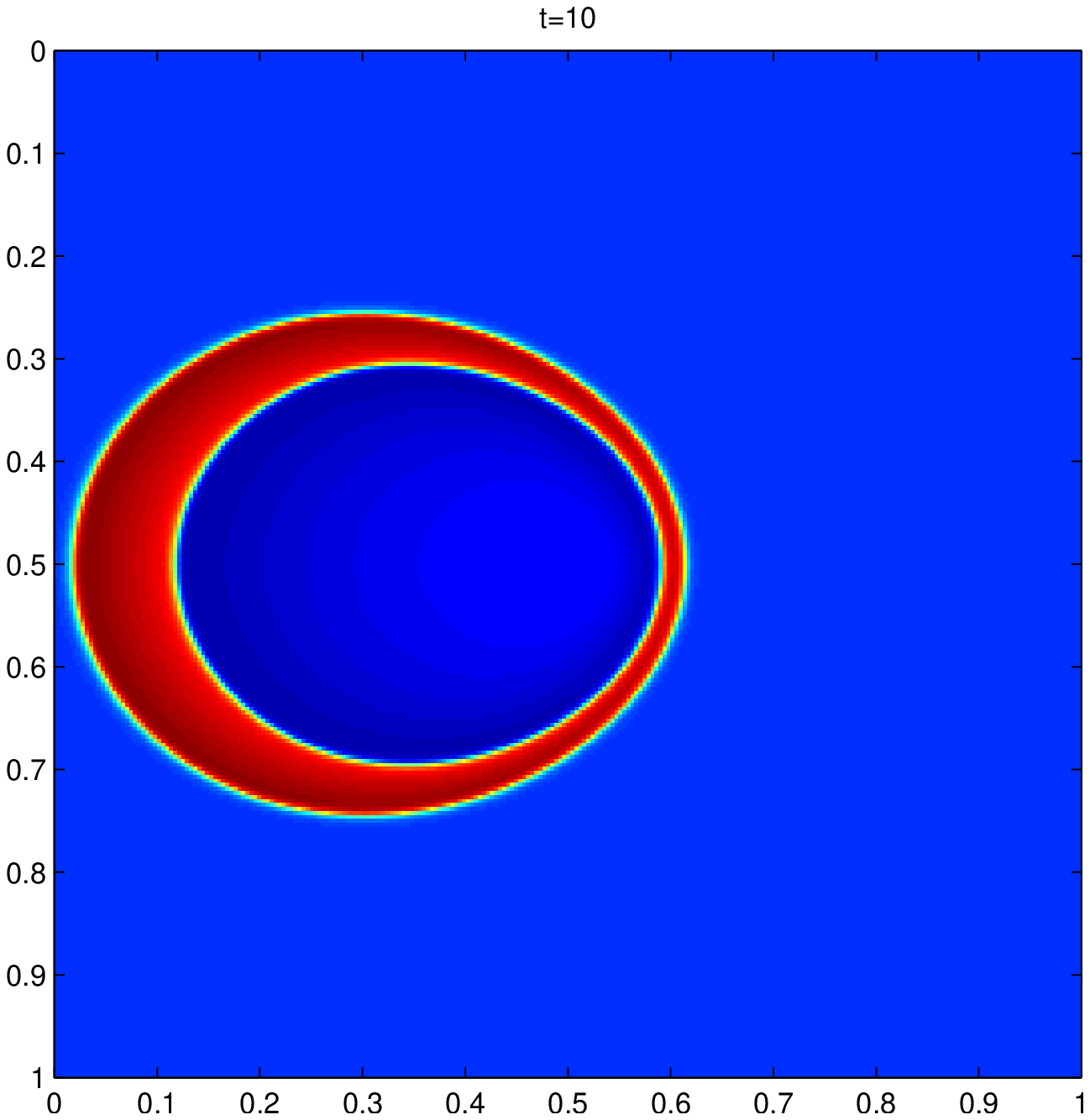}}\\\vspace{-0.65cm}
  \subfigure{\includegraphics[width=0.48\textwidth]{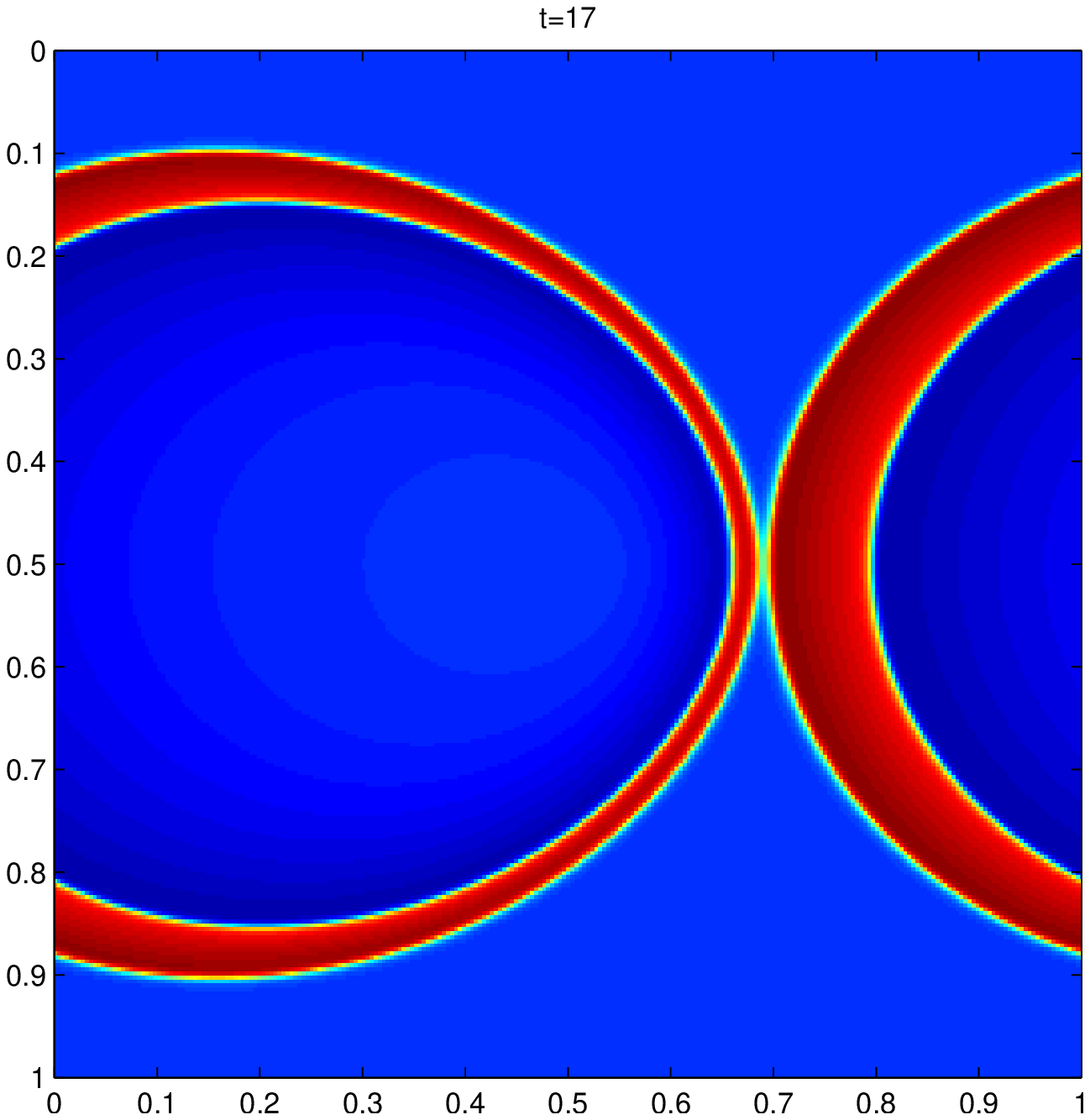}}
  \subfigure{\includegraphics[width=0.48\textwidth]{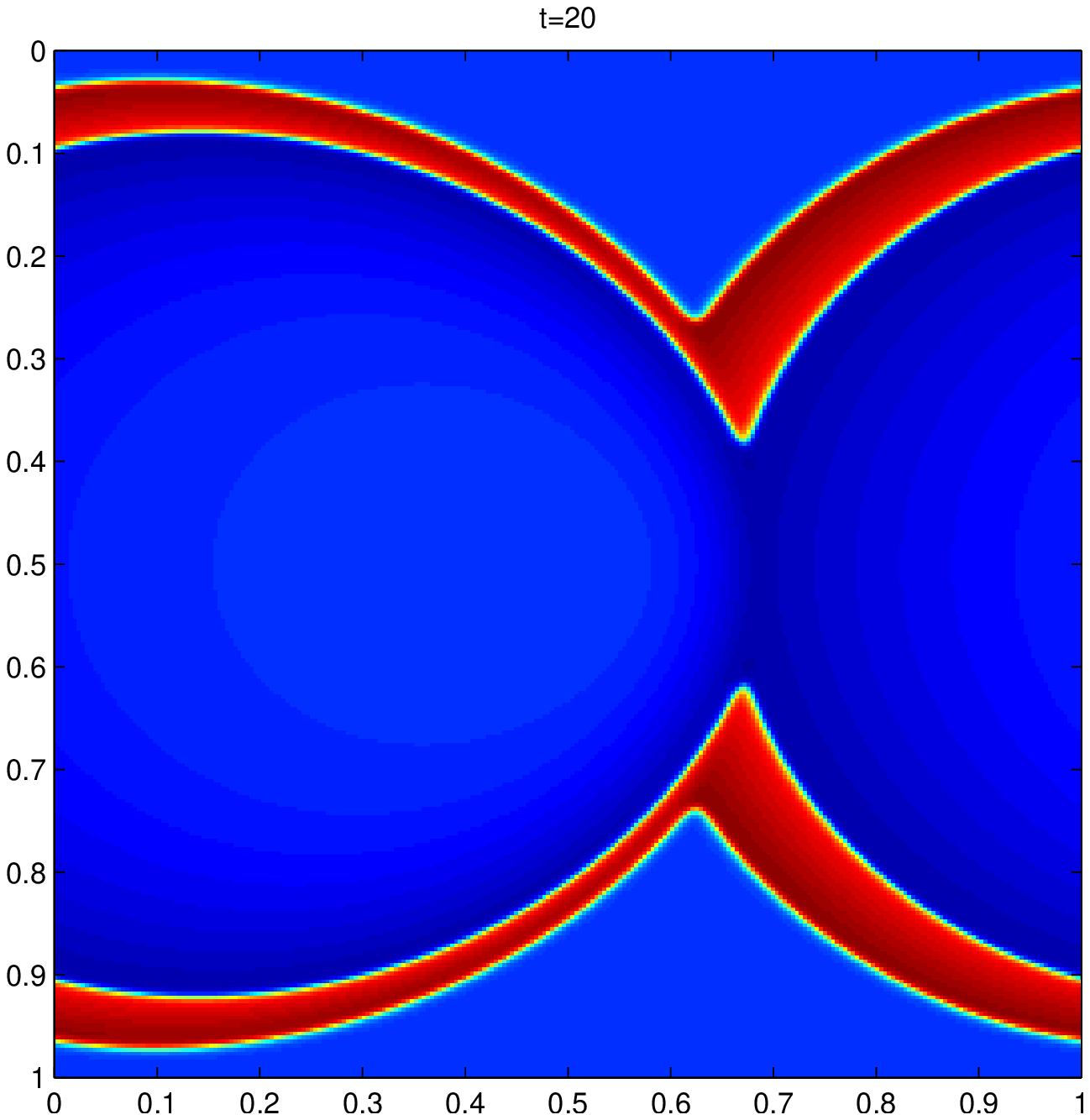}}
 \caption{\label{Fig:2D} 
 Two dimensional dynamics. Evolution of an initial stimulus by the discrete model of Theorem~\ref{theo:multiD} in a 
 $N=256\times 256$ lattice of neurons
 }
\end{figure}

\subsection{Pseudo-random connections}
While the models considered so far are fully deterministic, it is interesting to introduce some form of randomness and
monitor its effects. In the simplest form, this can be accomplished by perturbing the model considered above via 
 a (pseudo-)random removal of a fixed percentage of links among neurons. Connections to each neuron are turned-off with uniform distribution in the given percentage, independently of the other neurons; thus, the set ${\mathcal Q}(i)$ in~\eqref{Eq:FHNiQ}
does depend upon $i$, in a (pseudo-)random manner.
 
As an example, we keep the same parameters  $d=0.05$, $Q_N^D=\sqrt{2}$, $Q_N^C=2$ and  $n=256$, as well as the same initial datum as above, and we choose to turn 30\% of connections off. In Figure~\ref{Fig:2DRand}, the resulting dynamics at the same time instants as in Figure~\ref{Fig:2D} is shown. The random effects on the patterns are apparent. The reduction of active connections is reflected by a weaker propagation strength. The excitation front travels leftward only, with a lower speed than in the deterministic case. Furthermore, contours are irregular and, in some realizations not shown here, even disconnected. 
\begin{figure}[t]
\centering
  \subfigure{\includegraphics[width=0.48\textwidth]{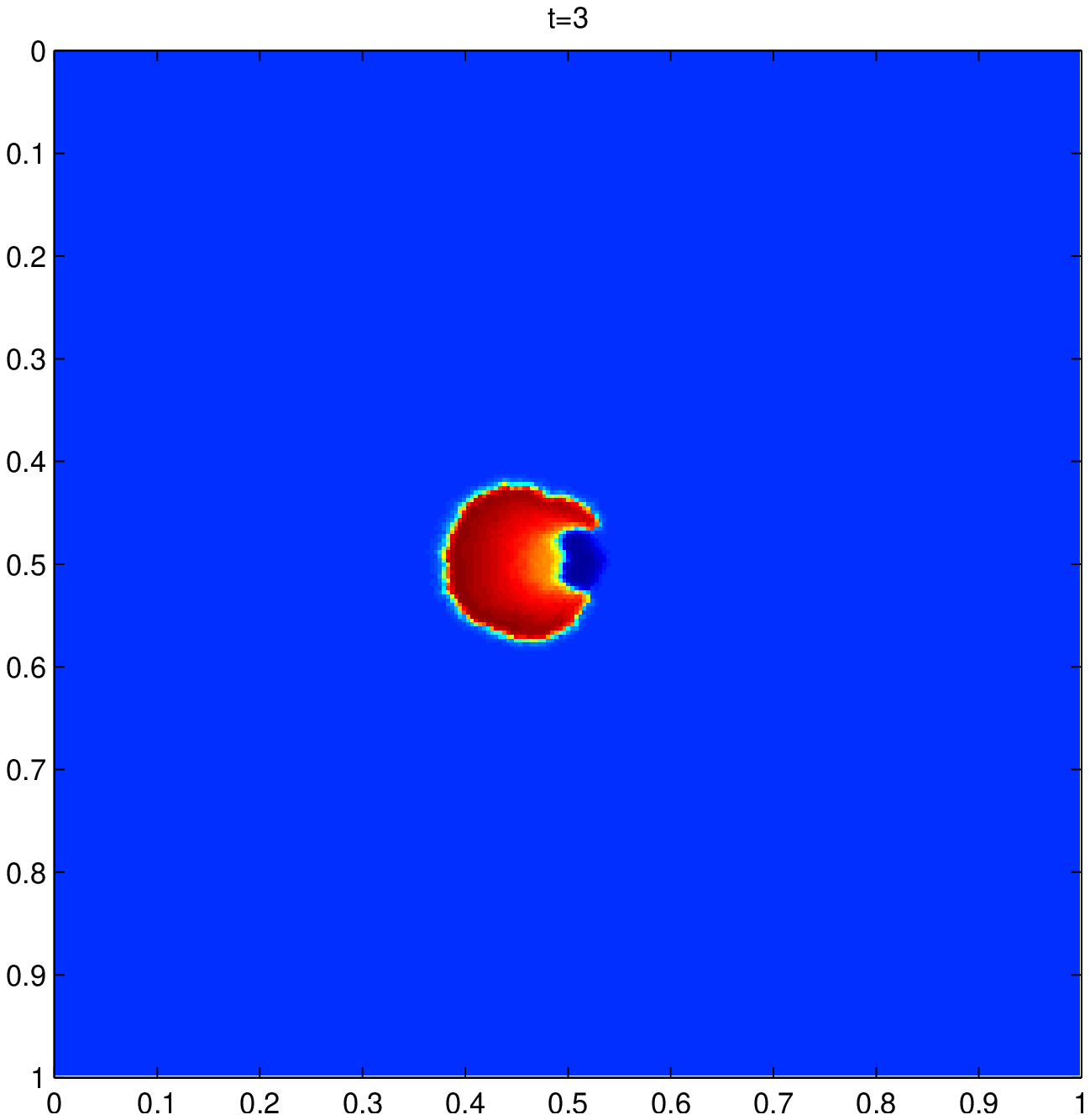}}
  \subfigure{\includegraphics[width=0.48\textwidth]{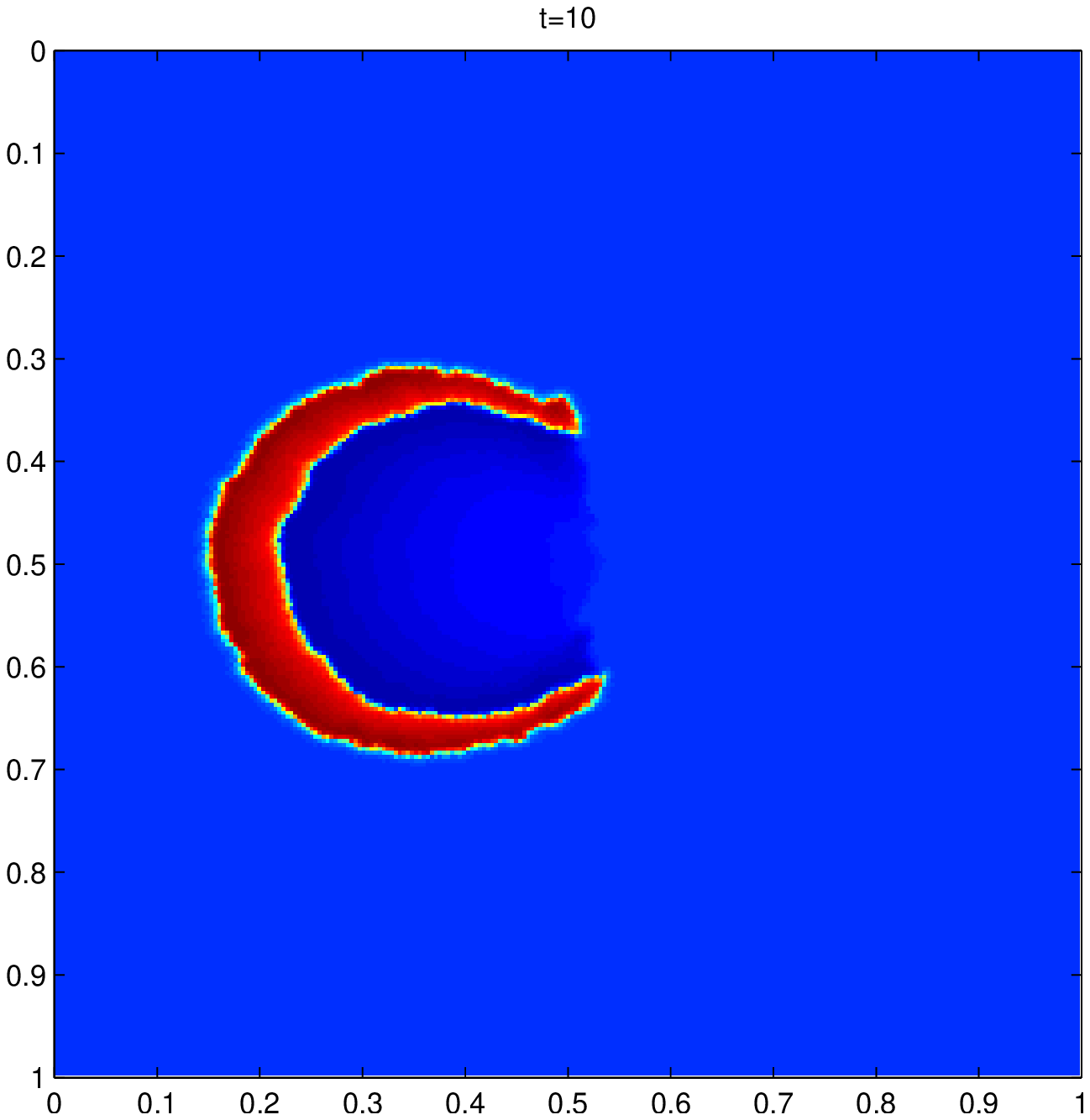}}\\\vspace{-0.65cm}
  \subfigure{\includegraphics[width=0.48\textwidth]{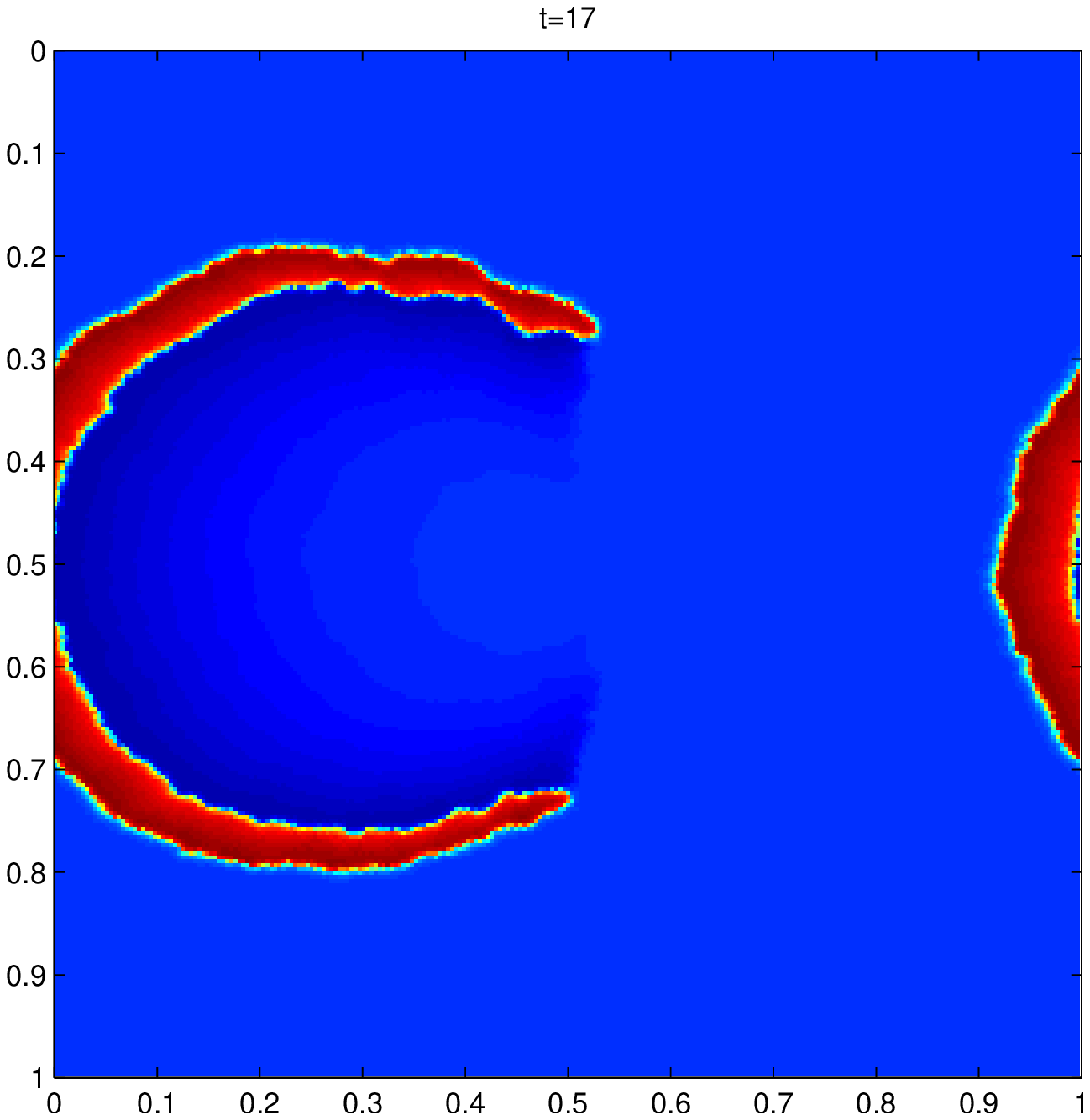}}
  \subfigure{\includegraphics[width=0.48\textwidth]{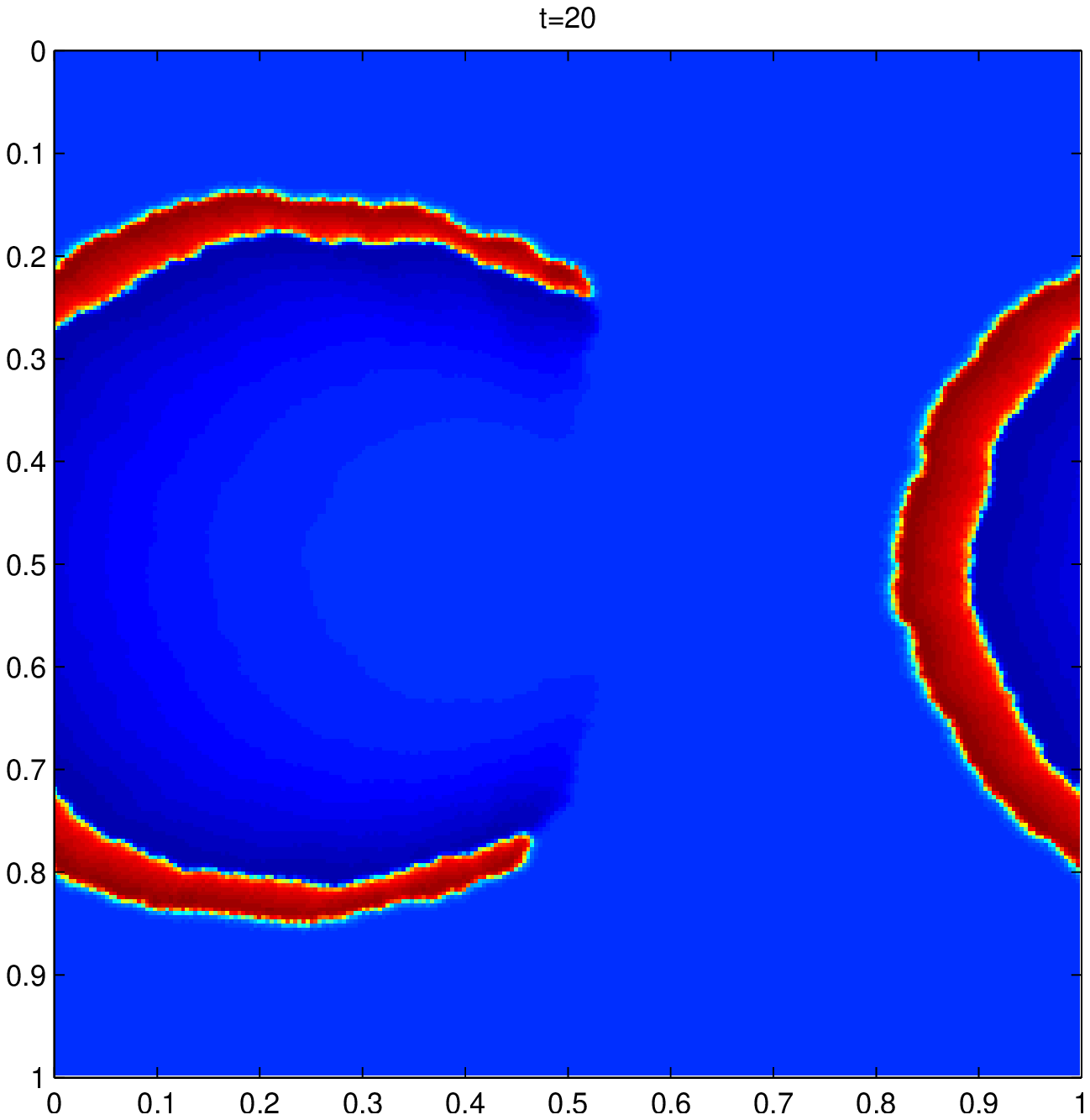}}
 \caption{\label{Fig:2DRand} Evolution as in Fig.~\ref{Fig:2D}, but with 30\% of the connections turned-off in a pseudo-random way
 }
\end{figure}

\section{Conclusions}
We have considered an idealized network, formed by $N$ neurons individually described by the FitzHugh-Nagumo equations
and connected by electrical synapses. The limit for $N \to \infty$ of the resulting discrete model has been thoroughly
investigated, with the aim of identifying a model for a continuum of neurons having an equivalent behaviour. Two strategies
for passing to the limit have been analyzed. A more conventional approach is based on a fixed nearest-neighbour
connection topology accompanied by a suitable scaling of the diffusion coefficients. We have devised a new approach, in which the number of connections to any given neuron varies with $N$ according to a precise law, which simultaneously
guarantees the non-triviality of the limit and the locality of neuronal interactions.  Both approaches yield in the limit 
a pde-based model, in which the distribution of action potential obeys a nonlinear reaction-convection-diffusion equation; 
convection accounts for the possible lack of symmetry in the connection topology. Several convergence issues are
discussed, both theoretically and numerically.

Based on the present study, in a forthcoming work we will consider more realistic models describing both electrical and
chemical synapses. The discrete models here analyzed will be coupled to models of chemical interactions within a population 
of excitatory/inhibitory neurons, such as those given in~\cite{Ermentrout}, eq.(9.6). Again, the focus will be on the limit process   
leading to coupled continuous models.

\section*{Acknowledgments} We would like to thank Piero Colli Franzone, Fabio Fagnani, Giovanni Naldi, Thierry Nieus, Luigi Preziosi and Giuseppe Savar\`e for enlightening discussions on various aspects of the present work.



\end{document}